\documentclass[11pt,letterpaper]{article}
\usepackage[margin=1in]{geometry}
\usepackage[T1]{fontenc}
\usepackage[utf8]{inputenc}
\usepackage{libertine}
\usepackage{microtype}
\usepackage{amsmath,amssymb,amsthm,mathtools}
\usepackage{enumitem}
\usepackage{xcolor}
\usepackage{tikz}
\usepackage[textsize=small]{todonotes}
\usepackage[font=small,labelfont=bf]{caption}
\usepackage[section]{placeins}
\usepackage{needspace}
\usepackage{tcolorbox}
\usepackage{graphicx,wrapfig}
\usepackage{booktabs}

\usepackage{dashrule}

\usepackage{lineno}

\usetikzlibrary{arrows.meta,calc,decorations.pathreplacing}
\usepackage[sort]{cite}
\usepackage[hidelinks]{hyperref}
\usepackage{aliascnt}
\usepackage[noabbrev,capitalize]{cleveref}
\hypersetup{
  pdftitle={Planar Metric Recognition is NP-Complete},
  pdfauthor={Anubhav Dhar; Debajyoti Kar; Sándor Kisfaludi-Bak; Hung Le; Bartłomiej Lewandowski; Karol Węgrzycki}
}

\definecolor{BrickRed}{rgb}{.72,0,0}
\def\EMPH#1{\emph{\textcolor{BrickRed} {#1}}}
\def\floor#1{\lfloor #1 \rfloor}

\newtheorem{theorem}{Theorem}[section]
\newaliascnt{lemma}{theorem}
\newtheorem{lemma}[lemma]{Lemma}
\aliascntresetthe{lemma}
\newaliascnt{observation}{theorem}
\newtheorem{observation}[observation]{Observation}
\aliascntresetthe{observation}
\newaliascnt{fact}{theorem}

\aliascntresetthe{fact}
\newaliascnt{corollary}{theorem}

\aliascntresetthe{corollary}
\newaliascnt{claim}{theorem}
\newtheorem{claim}[claim]{Claim}
\aliascntresetthe{claim}
\newaliascnt{question}{theorem}
\newtheorem{question}[question]{Question}
\aliascntresetthe{question}

\newaliascnt{proposition}{theorem}

\aliascntresetthe{proposition}
\theoremstyle{definition}
\newaliascnt{definition}{theorem}
\newtheorem{definition}[definition]{Definition}
\aliascntresetthe{definition}
\theoremstyle{remark}
\newaliascnt{remark}{theorem}

\aliascntresetthe{remark}
\crefname{observation}{Observation}{Observations}
\Crefname{observation}{Observation}{Observations}
\crefname{fact}{Fact}{Facts}
\Crefname{fact}{Fact}{Facts}
\crefname{claim}{Claim}{Claims}
\Crefname{claim}{Claim}{Claims}
\Crefname{question}{Question}{Questions}
\crefname{question}{Question}{Questions}

\usepackage{thm-restate}

\newcommand{\Oh}{\mathcal{O}}
\newcommand{\cT}{\mathcal{T}}

\newcommand{\ThreePartition}{\textsc{3-Partition}}

\newcommand{\eps}{\epsilon}
\newcommand{\Dr}{\mathfrak{D}}

\renewcommand{\leq}{\leqslant}
\renewcommand{\geq}{\geqslant}

\title{Metric Planarity Testing is Hard}
\date{}
\author{
  Anubhav Dhar\thanks{Max Planck Institute for Informatics, Saarbr\"ucken, Germany.
Supported by the Deutsche Forschungsgemeinschaft (DFG, German Research Foundation) grant number 559177164.
  }
  \and
  Debajyoti Kar\thanks{Department of Computer Science and Automation, Indian Institute of Science, Bengaluru, India. Supported by Google PhD Fellowship and a Microsoft Research India PhD Award.}
  \and
  S\'andor Kisfaludi-Bak\thanks{Department of Computer Science, Aalto University, Espoo, Finland. Supported by the Research Council of Finland, Grant 363444}
  \and
  Hung Le\thanks{University of Massachusetts Amherst, Amherst, MA, USA. Supported by the NSF CAREER award CCF-2237288 and the NSF grant CCF-2517033. }
  \and
  Bart{\l}omiej Lewandowski\footnotemark[1]
  \and
  Karol W\k{e}grzycki\footnotemark[1]
}

\begin{document}
\hypersetup{pageanchor=false}
\begin{titlepage}
\maketitle
\thispagestyle{empty}
\thispdfpagelabel{Title}

\begin{abstract}

A key milestone in algorithmic graph theory is Hopcroft and Tarjan’s 1974 linear-time algorithm for graph planarity testing. Linial and Rabinovich in a workshop in 2002~\cite{matouvsek2002open} asked for either a polynomial time algorithm or NP-hardness of \EMPH{metric planarity testing}:  Given a metric space $(X,d_X)$, test whether there exists an edge-weighted planar graph $G = (V\supseteq X, E,w)$ that realizes the distances in $X$: $d_G(x,y) = d_X(x,y)$ for every $x,y\in X$ This question arose in the context of research on metric embeddings and their applications in computer science since the mid-1990s. Over the past few decades, polynomial-time algorithms have been developed for several special cases.

In this work, we show that metric planarity testing is NP-hard. This sharply contrasts with graph planarity testing.  Our reduction yields a stronger result: under the Exponential Time Hypothesis (ETH), no $2^{o(n)}$-time algorithm for metric planarity testing exists, where $n$ is the number of points in the metric. 

To complement our ETH lower bound, we give two algorithmic results: (i) an algorithm for computing a metric planar embedding in $2^{\Tilde{O}(n^2)}$ time and (ii) a $(1+\epsilon)$-approximation algorithm in $2^{\Tilde{O}(n)}$ time for any fixed $\epsilon \in (0,1)$. The second algorithm produces a planar graph that preserves the input metric up to a factor of $(1+\epsilon)$.

\vfill

\paragraph{AI Disclosure.} Key technical ideas of this paper, including the NP-hardness reduction for metric planarity testing, ETH hardness lower bound for the 3-Partition problem, and the metric planar embedding algorithm with running time $2^{\Tilde{O}(n^2)}$, are due to \EMPH{ChatGPT-6 Astra}. The human authors claim no role in developing the technical ideas, other than merely verifying their correctness. (The $(1+\epsilon)$-approximation algorithm is not so hard to derive for experts in this area.) Instead, we mainly contributed to the presentation and asked questions. In particular, we contextualized the results and completely rewrote the paper, focusing on presenting the key new ideas in the proofs and attributing existing ideas to prior work (to the best of our knowledge). 

\end{abstract}
\end{titlepage}
\hypersetup{pageanchor=true}
\pagenumbering{arabic}

\section{Introduction}

The pioneering work of Linial, London, and Rabinovich~\cite{LLR95} viewed graphs as discrete metric spaces via shortest-path distances and brought key ideas from metric embeddings, such as the Johnson-Lindenstrauss lemma~\cite{JL84} and Bourgain's theorem~\cite{Bourgain85}, to solve fundamental graph problems.  Their result motivated a deeper understanding of the geometry of metric spaces for solving algorithmic problems. Recently, this algorithmic paradigm has been particularly productive in solving long-standing open problems in planar (and minor-free) graphs~\cite{FKS19,FL22,chang2023covering,chang2024shortcut,CCCLPP25,conroy2025protect}. The central idea is to take the shortest-path metrics of given planar graphs and devise metric decompositions with parameters substantially better than those achievable for general metrics; the topological properties of planar graphs were heavily used to facilitate these decompositions. Thus, the interplay between planar metrics and planar graphs is crucial in developing this idea.  

 We say that a metric space $(X,d_X)$ is a \EMPH{planar metric} if there exists an edge-weighted planar graph $G = (V,E,w)$ such that $X\subseteq V$ and for every two points $x\not= y \in X$, $d_X(x,y) = d_G(x,y)$. $G$ is called the \EMPH{planar embedding} of $(X,d_X)$.  Note that some vertices of $V$, called \EMPH{Steiner points}, might not correspond to any vertex in $X$. Steiner points are important for broadening the class of planar metrics: many metrics arising from non-planar graphs can also be planar. For example, the unweighted complete graph $K_n$ gives the uniform metric $(U,d_U)$, which is also a planar metric since it can be realized by a star graph:  points in $U$ are leaves of the star, and every edge of the star (adjacent to the center $c\not\in U$) has weight exactly $1/2$. Given this subtle definition, can we decide if a given metric is planar in polynomial time? This is precisely the question Linial and Rabinovic~\cite{matouvsek2002open} asked more than two decades ago.

\begin{question}[Linial and Rabinovic~\cite{matouvsek2002open}]\label{ques:main} Is it NP-complete or polynomial to decide if a finite metric space is planar?
\end{question}

A polynomial-time metric planarity testing algorithm could be influential for algorithmic metric space theory, as the Hopcroft-Tarjan planarity testing algorithm has been for algorithmic graph theory. And there has been various evidence supporting such an algorithm.

First, if one disallows Steiner vertices in the definition of planar metrics, then (non-Steiner) metric planarity testing can be easily reduced to graph planarity testing in polynomial time\footnote{Start with a graph $(X, \varnothing)$ with no edges and greedily add edges (whose weights are equal to the metric distances) to the graph in increasing order of the weights as long as the new edge is shorter than the distance between its endpoints in the current graph. Then $(X,d_X)$ admits a (non-Steiner) metric planar embedding if and only if the resulting graph is planar.}. As mentioned above, Steiner points broaden the class of planar metrics. There is no obvious reduction to graph planarity testing when Steiner vertices are allowed, which makes \Cref{ques:main} technically challenging. 

Second, there has been consistent progress towards testing subclasses of planar metrics, including tree metrics~\cite{Zaretskii1965,buneman1971,Bandelt1990,Batagelj1990}, ultrametrics~\cite{Gordon1987}, cactus metrics~\cite{HHMM20}, outer-planar metrics~\cite{ChenPylyavskyyTan2026}, Okamura-Seymour metrics~\cite{hurkens1988tidy,CO20,ChenPylyavskyyTan2026}, and its directed versions~\cite{ChenTan2025,LeTiernoYang2026}. In many of these cases, there is a so-called \EMPH{four-point condition}: for any tuple of four points, if their distances satisfy certain inequalities, then we can conclude the input metric falls into a certain class. For example, a metric $(X,d_X)$ is a tree metric if every four points satisfy the Buneman four-point condition~\cite{buneman1971}:
\begin{equation*}
    d_X(x_1,x_2) + d_X(x_3,x_4) \leq \max\{d_X(x_1,x_3)+d_X(x_2,x_4), d_X(x_1,x_4)+d_X(x_2,x_3)\}
\end{equation*}
A similar condition holds for Okamura-Seymour metrics with respect to a prescribed cyclic order of the terminals; see, e.g.,~\cite{ChenPylyavskyyTan2026}. While such a four-point condition often implies a polynomial-time testing algorithm\footnote{The case of tree metrics is easy: simply check every tuple of four vertices. However, for Okamura-Seymour metrics, the four-point condition depends on a certain cyclic ordering, and deriving that ordering requires substantially more effort, e.g.,~\cite{ChenTan2025}.}, there exists a subclass of planar metrics, namely outerplanar metrics~\cite{ChenPylyavskyyTan2026}, that has no such $O(1)$-point condition but still admits a polynomial-time testing algorithm. For planar metrics, an $O(1)$-point condition might not exist~\cite{ZihanPerCom}, but this does not rule out a polynomial-time testing algorithm.   

Besides theoretical importance, \Cref{ques:main} might also be of practical interest. It is closely related to Internet tomography~\cite{chung2001distance}, a central problem in networking, where the goal is to construct an unweighted graph (representing a subnetwork of the Internet) that realizes a distance matrix $(X,d_X)$ obtained by probing the Internet. Another example is constructing a (phylogenetic) tree that realizes a set of distances, which is a central problem in computational biology, e.g.,~\cite{FKW93}.

\subsection{Our Results}

We answer \Cref{ques:main} affirmatively: metric planarity testing is NP-hard.

\begin{theorem}\label{thm:main-NP-hardes} Metric planarity testing is NP-hard. 
\end{theorem}

Our early attempt to prove \Cref{thm:main-NP-hardes} was a reduction from a known NP-hard problem in planar graphs, for example, Monotone Planar 3-SAT~\cite{deBergKhosravi2010} or Planar Cubic Hamiltonian Cycle~\cite{GareyJohnsonTarjan1976}. The advantage of starting with a planar problem seems obvious: one might augment the input planar graph with local gadgets to get another (edge-weighted) planar graph, from which one can construct a planar metric for a given set of chosen terminals. The key difficulty is in setting up distances whose planar realizability encodes the source solution, while controlling all possible realizing graphs.  

\begin{figure}[!htb]
\centering
\begingroup
\definecolor{PMRblue}{RGB}{65,156,220}
\definecolor{PMRpurple}{RGB}{143,50,139}
\definecolor{PMRorange}{RGB}{230,158,53}
\definecolor{PMRgray}{RGB}{235,235,235}

\resizebox{0.9\linewidth}{!}{%
\begin{tikzpicture}[
  x=1cm,y=1cm,font=\large,
  line cap=round,line join=round,
  frame edge/.style={
    draw=black,line width=.8pt
  },
  major vertex/.style={
    circle,draw=black!80,fill=PMRgray,
    line width=.8pt,minimum size=13pt,inner sep=0pt
  },
  terminal/.style={
    circle,draw=none,fill=black,
    minimum size=4.5pt,inner sep=0pt
  },
  item vertex/.style={
    circle,draw=PMRpurple,fill=white,
    line width=1pt,minimum size=6.5pt,inner sep=0pt
  },
  crossing vertex/.style={
    rectangle,draw=black,fill=white,
    line width=.8pt,minimum size=4.8pt,inner sep=0pt
  },
  spine/.style={
    draw=PMRblue,line width=1.4pt
  },
  item path/.style={
    draw=PMRpurple,line width=.85pt,
    dash pattern=on 2.6pt off 2.3pt
  },
  length label/.style={
    PMRblue,fill=white,inner sep=.6pt,
    font=\scriptsize
  },
  transition/.style={
    draw=PMRorange,line width=1.8pt,
    -{Stealth[length=7pt,width=7pt]}
  }
]


\newcommand{\PMRoverviewSetup}{%
  \coordinate (a) at (-1.7,0);
  \coordinate (b) at ( 1.7,0);
  \coordinate (s) at (0, 2);
  \coordinate (t) at (0,-2);
  \coordinate (o) at (0,0);
}

\def\PMRoverviewSpine{%
  (s) .. controls (-5.4,2.8) and (-5.4,-2.8) .. (t)
}

\newcommand{\PMRoverviewBase}{%
  \foreach \u/\v in {a/s,s/b,b/t,t/a,o/a,o/b,o/s,o/t}{
    \draw[frame edge] (\u)--(\v);

    \foreach \r in {.25,.50,.75}{
      \node[terminal] at ($(\u)!\r!(\v)$) {};
    }
  }

  \draw[spine] \PMRoverviewSpine;
}

\newcommand{\PMRoverviewFrameVertices}{%
  \foreach \v in {a,b,s,t,o}{
    \node[major vertex] at (\v) {};
  }
}


\begin{scope}
  \draw[PMRpurple,line width=1.2pt]
    (0,-2.15) rectangle (1.05,4.00);

  \foreach \y in {2.95,1.90,.85,-1.10}{
    \draw[PMRpurple,line width=1.2pt]
      (0,\y)--(1.05,\y);
  }

  \node[PMRpurple] at (.525,3.475) {$x_1$};
  \node[PMRpurple] at (.525,2.425) {$x_2$};
  \node[PMRpurple] at (.525,1.375) {$x_3$};
  \node[PMRpurple] at (.525,-.125) {$\vdots$};
  \node[PMRpurple] at (.525,-1.625) {$x_{3m}$};

  \node[PMRpurple,font=\normalsize]
    at (.525,-2.65)
    {$\displaystyle\sum_{i=1}^{3m}x_i=mB$};

  \node at (.525,-3.70) {(a) $3m$ items};
\end{scope}

\draw[transition]
  (1.48,.45)
  .. controls (1.83,.96) and (2.28,.96) ..
  (2.66,.61);


\begin{scope}[xshift=7.2cm]
  \PMRoverviewSetup

  \coordinate (uone) at (-2.75,1.05);
  \coordinate (vone) at (-3.35,2.52);
  \coordinate (utwo) at (-1.50,1.25);
  \coordinate (vtwo) at (-1.30,3.00);

  \path \PMRoverviewSpine
    coordinate[pos=.06] (qone)
    coordinate[pos=.27] (qprev)
    coordinate[pos=.34] (qj)
    coordinate[pos=.305] (binmid);

  \PMRoverviewBase

  \foreach \r in {.06,.13,.20,.27,.34,.41,.50,.59,.68,.77,.86,.94}{
    \path \PMRoverviewSpine
      node[pos=\r,terminal] {};
  }

  \draw[item path] (a)--(uone);
  \draw[item path] (a)--(utwo);

  \draw[item path]
    (vone)
    .. controls (-.60,5.00) and (3.30,3.70) ..
    (b);

  \draw[item path]
    (vtwo)
    .. controls (-.10,4.40) and (2.80,3.10) ..
    (b);

  \draw[
    item path,
    preaction={draw=white,line width=3pt,solid}
  ]
    (uone)--(vone);

  \draw[
    item path,
    preaction={draw=white,line width=3pt,solid}
  ]
    (utwo)--(vtwo);

  \PMRoverviewFrameVertices

  \foreach \v in {uone,vone,utwo,vtwo}{
    \node[item vertex] at (\v) {};
  }


  \node at (-1.94,-.35) {$a$};
  \node at ( 1.96,-.35) {$b$};
  \node at ( .02, 2.42) {$q_0$};
  \node at ( .02,-2.43) {$q_m$};
  \node at (-.77,2.43) {$q_1$};

  \node[font=\normalsize]
    at (-3.61,1.67) {$q_{j-1}$};

  \node[font=\normalsize]
    at (-4.05,1.11) {$q_j$};

  \node at (-2.73,.73) {$u_1$};
  \node at (-3.48,2.84) {$v_1$};
  \node at (-1.18,1.25) {$u_2$};
  \node at (-1.20,3.34) {$v_2$};

  \node[PMRpurple]
    at (-2.06,1.18) {$\cdots$};

  \node[PMRpurple,font=\small]
    at (-2.42,.36) {$2x_1$};

  \node[PMRpurple,rotate=20,font=\normalsize]
    at (-2.50,3.63) {$i$-th tem path};

  \node[PMRblue,align=center,font=\normalsize]
    at (-4.45,2.40) {$j$-th\\bin path};

  \draw[
  PMRblue,line width=.8pt,
  -{Stealth[length=4pt]}
]
  (-4.38,1.97)
  .. controls (-4.55,1.54) and (-4.10,1.21) ..
  ($(binmid)+(-.10,0)$);

  \node at (-1.25,-3.70) {(b) Item paths and bin paths};
\end{scope}

\draw[transition]
  (9.43,.45)
  .. controls (9.80,.96) and (10.23,.96) ..
  (10.62,.61);


\begin{scope}[xshift=15.8cm]
  \PMRoverviewSetup

  \coordinate (ui) at (-1.933,1.267);
  \coordinate (vi) at (-2.433,2.700);
  \coordinate (uk) at (-2.400,.967);
  \coordinate (vk) at (-3.400,2.233);
  \coordinate (ul) at (-2.867,.480);
  \coordinate (vl) at (-4.167,1.773);

  \path \PMRoverviewSpine
    coordinate[pos=.10] (qprev)
    coordinate[pos=.16] (zi)
    coordinate[pos=.24] (zk)
    coordinate[pos=.32] (zl)
    coordinate[pos=.40] (qj);

  \PMRoverviewBase

  \foreach \r in {.10,.40,.55,.73,.89}{
    \path \PMRoverviewSpine
      node[pos=\r,terminal] {};
  }

  \foreach \u/\z/\v in {ui/zi/vi,uk/zk/vk,ul/zl/vl}{
    \draw[item path] (a)--(\u)--(\z)--(\v);
  }

  \draw[item path]
    (vi)
    .. controls (-.80,3.80) and (2.867,3.00) ..
    (b);

  \draw[item path]
    (vk)
    .. controls (-2.867,4.267) and (3.333,3.933) ..
    (b);

  \draw[item path]
    (vl)
    .. controls (-4.00,5.20) and (4.333,4.60) ..
    (b);

  \PMRoverviewFrameVertices

  \foreach \v in {ui,vi,uk,vk,ul,vl}{
    \node[item vertex] at (\v) {};
  }

  \foreach \v in {zi,zk,zl}{
    \node[crossing vertex] at (\v) {};
  }


  \node at (-1.94,-.35) {$a$};
  \node at ( 1.96,-.35) {$b$};

  \node[font=\normalsize]
    at (-1.15,2.37) {$q_{j-1}$};

  \node[font=\normalsize]
    at (-4.30,.55) {$q_j$};

  \node[font=\normalsize] at (-1.62,1.29) {$u_i$};
  \node[font=\normalsize] at (-2.45,2.97) {$v_i$};

  \node[font=\normalsize] at (-2.73,.85) {$u_k$};
  \node[font=\normalsize] at (-3.65,2.43) {$v_k$};

  \node[font=\normalsize] at (-3.06,.23) {$u_\ell$};
  \node[font=\normalsize] at (-4.47,1.98) {$v_\ell$};


  \node[PMRpurple,font=\normalsize]
    at (-1.30,-2.65) {$x_i+x_k+x_\ell=B$};

  \node at (-1.25,-3.70) {(c) a triple in a bin};
\end{scope}

\pgfresetboundingbox
\path[use as bounding box]
  (-.55,-4.00) rectangle (18.20,4.25);

\end{tikzpicture}%
}
\endgroup

\caption{%
  Overview of the reduction from \ThreePartition.
  (a) The $3m$ items have total size $mB$.
  (b) A graph $H$ constructed from the input of \ThreePartition. Items are encoded by purple item paths; the $m$ consecutive blue segments encode the $m$ bins. Two item paths are shown in the figure. The output metric $(X,d_X)$ of the reduction is the shortest-path metric of $H$. 
  (c) An idealized planar embedding of $(X,d_X)$ from subfigure (b). A triple $(x_i,x_{k}, x_{\ell})$ summing to $B$ assigned to the $j$-th bin corresponds to the three item paths intersecting the $j$-th bin path (from $q_{j-1}$ to $q_j$). This figure does not show the precise edge weight assignments. 
}
\label{fig:reduction-overview}
\end{figure}
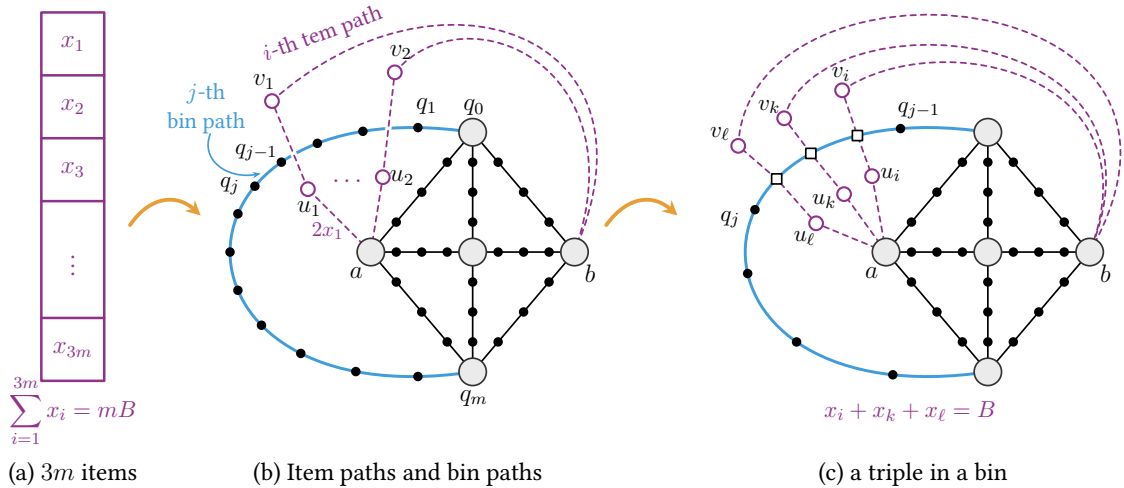

It turns out that the working reduction for proving \Cref{thm:main-NP-hardes} comes from the \ThreePartition~ problem, which has little to do with planarity. In this problem, we are given $3m$ positive integers $x_1,x_2,\ldots, x_{3m}$, called \EMPH{items}, and an integer $B$ sufficiently larger than $1$ such that:
\begin{equation}
    \sum_{i=1}^{3m}x_{i} = mB, \text{ and } B/4 < x_i < B/2 \quad \forall~ i \in [3m]~.
\end{equation}
The task is to determine whether the set of items can be partitioned into $m$ \EMPH{bins}, each containing exactly 3 items that sum to $B$. The reduction takes an input of the \ThreePartition~ problem and produces a metric space with $n = O(m)$ points with integer distances at most $O(mB)$. It is known that \ThreePartition~ is NP-hard even when $B = m^{O(1)}$~\cite{GareyJohnson1979}, i.e., it is strongly NP-hard, which implies the NP-hardness in \Cref{thm:main-NP-hardes}. \Cref{sec:reduction} describes the reduction in detail, which is quite simple. A short summary is: we can (i) encode $3m$ items by $3m$ paths in the planar embedding, called item paths, (ii) encode $m$ bins by $m$ other (bin) paths, and (iii) show that item $x_i$ is placed in the $j$-th bin if and only if its item path crosses the corresponding bin path in the embedding without changing the terminal distances. As a result, the existence of a planar embedding
implies that  \ThreePartition~ instance has a solution, and vice versa. See \Cref{fig:reduction-overview} for a simplified view.

Since the reduction preserves the number of items up to a constant factor, we could potentially obtain a stronger running-time lower bound under the Exponential Time Hypothesis (ETH). In particular, if  \ThreePartition~ does not admit $2^{o(m)}B^{O(1)}$ running time under ETH, then we can immediately rule out $2^{o(n)}$ running time for metric planarity testing.  Unfortunately,  this running time is ruled out under ETH only for $k$-\textsc{Partition} for $k\geq 4$ ~\cite[Corollary 3.2]{Jansen2016}; for \ThreePartition, the state of the art is $2^{o(\sqrt{m})}B^{O(1)}$ ETH lower bound~\cite{Jansen2016,GareyJohnson1979}. (We could change our reduction to be from  $4$-\textsc{Partition} and then apply the black box $2^{\Omega(m)}B^{O(1)}$ ETH lower bound~\cite{Jansen2016}, but this approach is arguably less elegant.) In Appendix~\ref{app:three-partition}, we provide a $2^{\Omega(m)}B^{O(1)}$ lower bound for \ThreePartition~ by a direct reduction from the \textsc{3-Dimensional Matching} problem. 

\begin{restatable}{theorem}{ETHThreePartition}\label{thm:ETH-3Partition} Assuming the Exponential Time Hypothesis (ETH), \ThreePartition\ on $3m$ numbers has no algorithm with running time $2^{o(m)}|B|^{O(1)}$, where $|B|$ is the encoding length. This holds even for instances with $B=m^{O(1)}$.
\end{restatable}

As a corollary of \Cref{thm:ETH-3Partition} and the reduction in our proof of \Cref{thm:main-NP-hardes}, we rule out a $2^{o(n)}$-time algorithm for metric planarity testing under ETH.

\begin{restatable}{corollary}{ETHhardnessPlanar}\label{thm:ETH-hardness}
There exists a constant $c$ such that given any metric space $(X,d_X)$ with $n$ vertices and $d_X(x,y) \in \{1,2,\ldots,n^{c}\}$ for $x\not= y \in X$, any algorithm testing planarity of $(X,d_X)$ requires $2^{\Omega(n)}$ time unless ETH fails.
\end{restatable}

We complement the ETH lower bound with two exponential-time algorithms for metric planar embedding. It is known that any planar metric with $n$ points has a planar embedding with $O(n^4)$ vertices~\cite{KNZ14}. This yields a simple $2^{\Tilde{O}(n^4)}$-time algorithm: enumerate all planar graphs with $O(n^4)$ vertices and solve a linear program for each candidate to find edge weights that realize the input metric. This approach uses only the sparsity of planar graphs. We exploit the structure of planar embeddings further to enumerate candidate planar embeddings in $2^{\Tilde{O}(n^2)}$ time. The key is the \EMPH{rotation system} of a simple drawing of the complete graph: the cyclic order of the curves incident to each terminal. The rotation system uniquely determines which pairs of curves cross~\cite{pach2006many,kynvcl2009enumeration}. A naive estimate bounds the number of rotation systems by at most $((n-2)!)^n=2^{O(n^2\log n)}$, so we can enumerate them and use linear programming to check whether the resulting candidates admit edge weights that realize the input metric. However, not all of these rotation systems correspond to simple drawings. Kyn\v{c}l~\cite{kynvcl2013improved} bounded the number of feasible rotation systems by $2^{n^2\alpha(n)^{O(1)}}$, where $\alpha(n)$ is the inverse Ackermann function. While his proof uses VC-dimension arguments and is existential, we show that these rotation systems can also be enumerated within the same asymptotic time bound. This gives the following result.

\begin{theorem}\label{thm:algorithm} Given any metric space $(X,d_X)$ with $n$ points whose distances are in $\{1,2,\ldots,\Delta\}$ for some integer $\Delta > 0$, we can construct a planar embedding for  $(X,d_X)$, if one exists, in time $2^{n^2\alpha(n)^{O(1)}}(\log \Delta)^{O(1)}$, where $\alpha(n)$ is the  inverse Ackermann function. 
\end{theorem}

We leave closing the gap between the $2^{\Omega(n)}$ ETH lower bound and the running time in \Cref{thm:algorithm} as an open problem. Interestingly, if we apply the algorithm in \Cref{thm:algorithm} to the hard instance in our ETH-hardness reduction, it runs in $2^{\Tilde{O}(m)}$ time for a very specific reason: we only need to enumerate the ordering of the item paths around \EMPH{a single vertex}, specifically vertex $a$ in \Cref{fig:reduction-overview}(b). If one can cook up a hard instance where we really have to enumerate the ordering of paths of $\Omega(n)$ vertices, then one could likely get a $2^{\Omega(n^2)}$ ETH lower bound  (to match the running time in \Cref{thm:algorithm}), though this is wildly conjectural.

Finally, if we allow a $(1+\epsilon)$-approximation, we can get a singly exponential-time algorithm. In this algorithm, the input metric $(X,d_X)$ need not be planar, as long as a $(1+\epsilon/2)$-approximation exists.

\begin{theorem}
\label{thm:approximate-algo}
    For any $\epsilon\in (0,1)$, and any given metric space $(X,d_X)$ with $n$ points whose distances are in $\{1,2,\ldots,\Delta\}$ for some integer $\Delta > 0$,  we can find in $2^{n(\log n/\epsilon)^{O(1)}}(\log \Delta)^{O(1)}$ time an edge-weighted planar graph $G = (V\supseteq X, E, w)$ such that:
    \begin{equation}
        d_X(x,y) \leq d_G(x,y) \leq (1+\eps)d_X(x,y)\qquad x\not=y \in X,
    \end{equation}
    if such a graph with $(1+\eps/2)$-approximation exists.
\end{theorem}

The idea of the algorithm is very simple: it is known that any planar metric $(X,d_X)$ has a planar embedding that preserves distances up to a factor of $1+\epsilon$ of size $n(\log n/\eps)^{O(1)}$ \cite{chang2022almost}. Thus, we can enumerate this emulator in time $2^{n(\log n/\epsilon)^{O(1)}}$.

\paragraph{On notation.} Here we only use standard graph terminology. In particular, $\pi_G(u,v)$ denotes a shortest path from $u$ to $v$ in $G$, and $d_G(u,v)$ denotes the distance between $u$ and $v$ in $G$. We will use the terms weight and length interchangeably.

\section{ETH Hardness of Metric Planarity Testing}\label{sec:reduction}

We reduce from the \ThreePartition~ problem: given a set of $3m$ items $\{x_1,\ldots, x_{3m}\}$ and an integer $B$ (sufficiently larger than $1$), we construct an edge-weighted (non-planar) graph $H$ in three steps. The output metric $(X,d_X)$ is the shortest-path metric of $H$.  Let:
\begin{equation}\label{eq:def-L}
     L = m(4B+2)
\end{equation}
Recall that the total sum of all the numbers is $mB$, so $L$ is larger than the sum. For now, one should think of $L$ as a number sufficiently larger than $m$ and $B$ (and all the $x_i$'s).  

\begin{figure}[t]
  \centering
  \begingroup
  \definecolor{PMRblue}{RGB}{65,156,220}
  \definecolor{PMRpurple}{RGB}{143,50,139}
  \definecolor{PMRorange}{RGB}{230,158,53}
  \definecolor{PMRgray}{RGB}{235,235,235}

  \resizebox{\linewidth}{!}{%
  \begin{tikzpicture}[
      x=1cm,y=1cm,font=\large,
      line cap=round,line join=round,
      frame edge/.style={
        draw=black,line width=.85pt
      },
      major vertex/.style={
        circle,draw=black!80,fill=PMRgray,
        line width=.8pt,minimum size=12pt,inner sep=0pt
      },
      subdivision vertex/.style={
        circle,fill=black,draw=none,
        minimum size=4.5pt,inner sep=0pt
      },
      item vertex/.style={
        circle,draw=PMRpurple,fill=white,
        line width=1.05pt,minimum size=6.5pt,inner sep=0pt
      },
      spine/.style={
        draw=PMRblue,line width=1.5pt
      },
      marker link/.style={
        draw=PMRblue,line width=1.2pt,
        dash pattern=on 0pt off 3.4pt
      },
      item edge/.style={
        draw=PMRpurple,line width=.95pt,
        dash pattern=on 3.3pt off 3pt
      },
      edge label/.style={
        fill=white,inner sep=1.5pt
      },
      panel label/.style={
        font=\large,align=center
      }
    ]

    \newcommand{\PMRframe}[1]{%
      \coordinate (#1-a) at (-1.7,0);
      \coordinate (#1-b) at ( 1.7,0);
      \coordinate (#1-s) at (0, 2);
      \coordinate (#1-t) at (0,-2);
      \coordinate (#1-o) at (0, 0);

      \foreach \u/\v in {a/s,s/b,b/t,t/a,o/a,o/b,o/s,o/t}{
        \draw[frame edge] (#1-\u)--(#1-\v);
        \foreach \r in {.25,.50,.75}{
          \node[subdivision vertex]
            at ($(#1-\u)!\r!(#1-\v)$) {};
        }
      }

      \foreach \v in {a,b,s,t,o}{
        \node[major vertex] at (#1-\v) {};
      }

      \node[below left=4pt]  at (#1-a) {$a$};
      \node[below right=4pt] at (#1-b) {$b$};
      \node[above right=5pt] at (#1-o) {$o$};
    }

    \newcommand{\PMRspine}[1]{%
      \draw[spine,shorten <=6.3pt,shorten >=6.3pt]
        (#1-s)
        .. controls (-5.4,2.8) and (-5.4,-2.8) ..
        (#1-t)
        coordinate[pos=.065] (#1-qone)
        coordinate[pos=.135] (#1-qtwo)
        coordinate[pos=.370] (#1-qj)
        coordinate[pos=.775] (#1-bin);

      \foreach \r in
        {.065,.135,.21,.29,.37,.46,.55,.64,.73,.82,.91}{
        \path
          (#1-s)
          .. controls (-5.4,2.8) and (-5.4,-2.8) ..
          (#1-t)
          node[pos=\r,subdivision vertex] {};
      }

      \draw[marker link,shorten >=6.3pt]
        (#1-qj)--(#1-a);

      \draw[marker link,shorten >=6.3pt]
        (#1-qj)
        .. controls (-5.9,6.25) and (4.3,6.0) ..
        (#1-b);

      \node[left=5pt] at (#1-qj) {$q_j$};
      \node[PMRpurple,edge label,below left=3pt]
        at (#1-bin) {\small{$4B+2$}};
    }

    \draw[PMRorange!80,line width=.65pt]
      (2.75,-3.6)--(2.75,5.05);
    \draw[PMRorange!80,line width=.65pt]
      (10.35,-3.6)--(10.35,5.05);

    \begin{scope}
      \PMRframe{F}

      \node[above=6pt] at (F-s) {$s$};
      \node[below=6pt] at (F-t) {$t$};

      \node[PMRblue,edge label]
        at (-1.55,1.48) {$L+3$};
      \node[PMRblue,edge label]
        at (1.26,-1.47) {$L+3$};



      \node[PMRorange,edge label,below=5pt]
        at (-.65,0) {\small{$L+2$}};
      \node[PMRorange,edge label,below=5pt]
        at (.65,0) {\small{$L+2$}};

      \draw[
        -{Stealth[length=4pt]},
        PMRorange,line width=.8pt
      ]
        (1.33,1.53)
        --($(F-s)!.37!(F-b)+(.10,.06)$);

      \node[PMRpurple,anchor=west]
        at (1.10,1.87) {$\le B/2$};

    \end{scope}

    \node[panel label] at (0,-3.42)
      {(a) Frame $F$};

    \begin{scope}[xshift=7.4cm]
      \PMRframe{J}
      \PMRspine{J}

      \node[above=7pt] at (J-s) {$s=q_0$};
      \node[below=7pt] at (J-t) {$t=q_m$};
      \node[above=5pt] at (J-qone) {$q_1$};
      \node[above=5pt] at (J-qtwo) {$q_2$};

      \node[PMRblue,edge label]
        at (-2.50,.81) {$L+3$};
      \node[PMRblue,edge label]
        at (-2.88,4.59) {$L+3$};

      \node[PMRblue,align=center]
        at (-2.63,-2.64) {Spine $Q$};
    \end{scope}

    \node[panel label] at (6.62,-3.42)
      {(b) Graph $J$};

    \begin{scope}[xshift=15cm]
      \PMRframe{H}
      \PMRspine{H}

      \node[above=7pt] at (H-s) {$s$};
      \node[below=7pt] at (H-t) {$t$};

      \coordinate (H-ui) at (-2.75,1.05);
      \coordinate (H-vi) at (-3.35,2.52);
      \coordinate (H-uk) at (-1.50,1.25);
      \coordinate (H-vk) at (-1.30,3.00);

      \draw[item edge,shorten <=6.3pt]
        (H-a)--(H-ui);
        \node[PMRpurple,edge label]
             at (-2.60,.32) {$2x_i$};

      \draw[item edge,shorten <=6.3pt]
        (H-a)--(H-uk);

      \draw[item edge,shorten >=6.3pt]
        (H-vi)
        .. controls (-.60,5.00) and (3.30,3.70) ..
        (H-b);
        \node[PMRpurple,edge label]
        at (-1.90,3.92) {$2x_i$};

      \draw[item edge,shorten >=6.3pt]
        (H-vk)
        .. controls (-.10,4.40) and (2.80,3.10) ..
        (H-b);

      \draw[
        item edge,
        preaction={draw=white,line width=3.8pt,solid}
      ]
        (H-ui)--(H-vi)
        node[
          pos=.58,left=3pt,
          PMRpurple,edge label
        ] {\small{$2L+4$}};

      \draw[
        item edge,
        preaction={draw=white,line width=3.8pt,solid}
      ]
        (H-uk)--(H-vk);

      \node[
        item vertex
      ] at (H-ui) {};

      \node[
        item vertex
      ] at (H-vi) {};

      \node[
        item vertex
      ] at (H-uk) {};

      \node[
        item vertex
      ] at (H-vk) {};
    \node at (-2.30,1.05) {$u_i$};
    \node at (-2.85,2.50) {$v_i$};
    \node at (-1.05,1.32) {$u_1$};
    \node at (-0.95,2.75) {$v_1$};

    \end{scope}

    \node[panel label] at (14.22,-3.42)
      {(c) Graph $H$};

  \end{tikzpicture}%
  }
  \endgroup

  \caption{%
    (a) The frame $F$ is a subdivided wheel: $d_F(o,a) = d_F(o,b) = L+2$ while other distances $d_F(a,s) = d_F(a,t) = d_F(s,b) = d_F(b,t) = d_F(o,s) = d_F(o,t)= L+3$; each subdivision edge has a positive integer length at most $B/2$.
    (b) Adding the spine $Q=(q_0=s,\ldots,q_m=t)$ of total length $L$ where each edge has length $4B+2$; two edges $(a,q_j)$ and $(b,q_j)$ of length $L+3$ each shown in the figure. Clearly $J$ is planar. (c) Adding one item path $(a,u_i, v_i, b)$ per item $x_i$; the three edge lengths are  $2x_i,2L+4,2x_i$. The graph $H$ is non-planar; the figure shows crossings between item paths and the spine, but there is no vertex at each crossing.
  }
  \label{fig:frame-spine-items}
\end{figure}
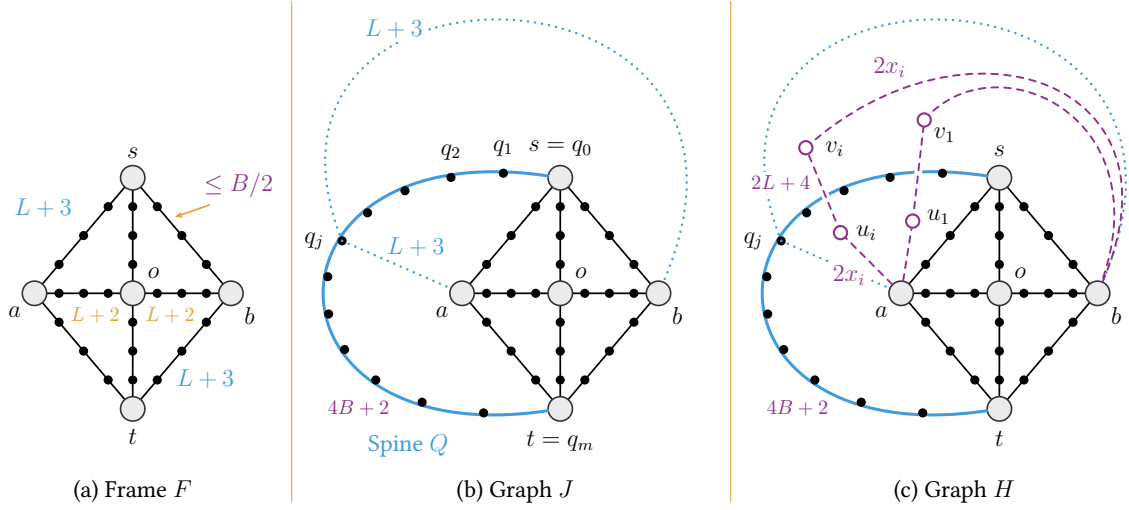

\hdashrule{0.95\linewidth}{1pt}{2pt}

\paragraph{Step 1: The frame $F$.} First, construct a graph on vertex set $\{a,b,s,t, o\}$ with 8 edges, where $\{a,s,b,t\}$ forms a cycle and $o$ has 4 edges to all other four vertices, i.e., a wheel graph. Every edge has length $L+3$, except for $(o,a)$ and $(o,b)$, which have length $L+2$ only.  For every edge, we subdivide it into a path such that each edge on the path has an integer length $\floor{B/2}$; the last edge could have length smaller than $\floor{B/2}$, and in this case we subdivide so that this edge is not incident to either $a$ or $b$. Each path has $O(L/B) = O(m)$ edges.  The resulting graph is called the \EMPH{frame}, denoted by $F$. See \Cref{fig:frame-spine-items}(a). 

\paragraph{Step 2: The spine $Q$ and graph $J$.} We add to $F$ a path $Q = (q_0 = s, q_1,\ldots, q_m = t)$ of $m$ edges where each edge has length $4B+2$, so that the total length of the path is exactly $m(4B+2) = L$ by \Cref{eq:def-L}. (At a high level, each edge plays the role of a bin in the solution of \ThreePartition; we will make this intuition more precise later.)  Finally, for each internal spine vertex $q_j \in Q\setminus \{s,t\}$, we add two edges $(a,q_j)$ and $(q_j,b)$, each of length exactly  $L+3$. Let $J$ be the resulting graph. Observe that $J$ is planar; see \Cref{fig:frame-spine-items}(b).  

Note that previously, the distance between $s$ and $t$ in the frame $F$ is $2(L+3)$. However, by attaching the spine $Q$ (of length $L$) to the frame, the distance between $s$ and $t$ is now $L$. Some distances between subdivided vertices in $F$ near $s$ and $t$ also get shorter, but it is not hard to show that the distances between $a,b,o$ and from $\{a,b,o\}$ to $\{s,t\}$ remain the same. 

\paragraph{Step 3: Item paths and graph $H$.} For each item $x_i$, we attach to $J$ a new path $(a,u_i,v_i,b)$ between $a$ and $b$, and set the weight of the three edges $(a,u_i)$, $(u_i,v_i)$ and $(v_i,b)$ to be $2x_i, 2L + 4, 2x_i$, respectively. This path has length exactly $2L + 4x_i + 4$, larger than the distance between $a$ and $b$ in $J$ (which is $2L+4$), so the distances between vertices in $J$ remain the same. Let $H$ be the resulting graph (with edge weights to be the assigned lengths). We note that $H$ is non-planar; see \Cref{fig:frame-spine-items}(c).  

\paragraph{The metric $(X,d_X)$.} The resulting metric is $(X,d_X) \coloneqq (V(H), d_H)$, the shortest-path metric of $H$. 

\hdashrule{0.95\linewidth}{1pt}{2pt}

The idea of the item paths is the following: roughly speaking, in the solution of the \ThreePartition~ problem, three items, say $x_i, x_k, x_{\ell}$, are put into the same bin, say the $j$-th bin \EMPH{if and only if} in the planar graph $G$ realizing $(V(H),d_H)$, the item paths corresponding to $x_i, x_k$ and $x_{\ell} $ intersect the edge $(q_{j-1}, q_{j})$ (at some vertex of $G$), and no other item path intersects $(q_{j-1}, q_{j})$. Here we implicitly assume that $G$ is a planarization of $H$, obtained by introducing a vertex for the crossing between each item path and the spine.  Indeed, in the yes case (completeness), $G$ is a planarization of $H$. The bulk of the soundness analysis is showing that $G$ is close to a planarization of $H$. Before going into the details, we collect several important facts about the distances in $F$, $J$, and $H$. (On a first reading, we suggest that readers skip this observation and refer back to it later.)

\begin{observation}\label{obs:distances-H} $|V(H)| = O(m)$ and:
\begin{align}
    d_F(a,b) = d_H(a,b) &= 2L+4  &\label{eq:a-b-dist}\\
    d_F(s,t) = 2L+6,~  d_H(s,t) &= L& \label{eq:s-t-dist}\\
    d_H(a, q_j) = d_H(b,q_j) &= L+3  &\qquad\forall j \in [0,m] \label{eq:ab-qj-dist}\\
    d_H(u_i, z) = 2x_i + d_J(z,a),~ d_H(v_i, z) &= 2x_i + d_J(z,b) & \qquad\forall z\in V(J) \text{ and } i \in [3m] \label{eq:ui-z-inJ}\\
    d_H(u_i,v_i) &= 2L+ 4  &\qquad\forall i \in [3m]\label{eq:ui-vi-dist}\\
    d_H(a,v_i) = d_H(b,u_i) &= 2L + 4 +2x_i  &\qquad\forall i \in [3m] \label{eq:a-vi-b-u_i-dist}\\
    d_H(u_i,v_k) &= 2L+4+2x_i + 2x_k   &\qquad\forall i\not=k \in [3m] \label{eq:u_i-v_k-dist}\\
    d_H(u_i, q_j) = d_H(v_i,q_j) &= L+3+2x_i  &\qquad\forall i \in [3m], j \in [0,m]  \label{eq:u_i-q_j-dist}\\
    d_H(q_{j}, q_k) = d_J(q_j,q_k) &=  |j-k|(4B+2) \leq L &\qquad \forall i, k \in [0,m]\label{eq:q_j-q_k-dist}
\end{align}
and by the choice of $L$, we have:
\begin{align}
    B/4 < x_i < B/2 \ll L \qquad\forall i \in [3m] \label{eq:xi-B-L}
\end{align}
\end{observation}
\begin{proof} The first three equations, namely \Cref{eq:a-b-dist,eq:s-t-dist,eq:ab-qj-dist} follow directly from the construction. Next, for every $z\in J$, there is a shortest path from $u_i$ to $z$ going through $a$ and then from $a$ to $z$ inside $J$. Since $(u_i,a)$ has weight $2x_i$, the total length of the path is $2x_i + d_J(z,a)$ as in \Cref{eq:ui-z-inJ}. The same argument applies to $d_H(v_i, z)$ by symmetry between $a$ and $b$.  

\Cref{eq:ui-vi-dist} follows from the fact that the edge $(u_i,v_i)$ of length $2L+4$ is the shortest path between $u_i$ and $v_i$; any other path must go through $a$ and $b$ and by \Cref{eq:a-b-dist} has length strictly larger than $2L+4$. For \Cref{eq:a-vi-b-u_i-dist}, notice that the shortest path between $a$ and $v_i$ is the subpath $(a,u_i,v_i)$ of the item path, which has length $2L+4+2x_i$; the same bound holds for $d_H(b,u_i)$ for the same reason.  For $d_H(u_i,v_k)$, one shortest path is from $u_i$  to $a$ (of length $2x_i$), to $u_k$ (of length $2x_k$), and finally to $v_k$ (of length $2L+4$), totaling $2L+4+2x_i + 2x_k$, giving \Cref{eq:u_i-v_k-dist}. For   $ d_H(u_i, q_j)$,  the shortest path is from $u_i$ to $a$ (of length $2x_i$), and then $a$ to $q_j$ (of length $L+3$), totaling $L+3+2x_i$ as in \Cref{eq:u_i-q_j-dist}. For \Cref{eq:q_j-q_k-dist}, the shortest path between $q_j$ and $q_k$ is the subpath $Q[q_j,q_k]$ of the spine, which has length exactly $|j-k|(4B+2)$. The last equation follows from the definition of $x_i$ and $L$.
\end{proof}

Our analysis in \Cref{subsec:completeness} and \Cref{subsec:soundness}  is to show:

\begin{lemma}\label{lm:reductio-overall}  The set of items $\{x_1,\ldots, x_{3m}\}$ admits a partition into $m$ triples, each summing up to exactly $B$, if and only if the metric $(V(H),d_H)$ is planar.
\end{lemma}

Given \Cref{lm:reductio-overall}, we can now prove our main theorem, restated below.

\ETHhardnessPlanar*
\begin{proof}
    By \Cref{obs:distances-H}, $|V(H)| = O(m)$. Thus, by \Cref{lm:reductio-overall}, an algorithm with running time $2^{o(n)}$ for metric planarity testing will imply an algorithm with running time $2^{o(|V(H)|)} = 2^{o(m)}$ for \ThreePartition~ when $B  = m^{O(1)}$, which contradicts~\Cref{thm:ETH-3Partition}.
\end{proof}

\subsection{Completeness}\label{subsec:completeness}

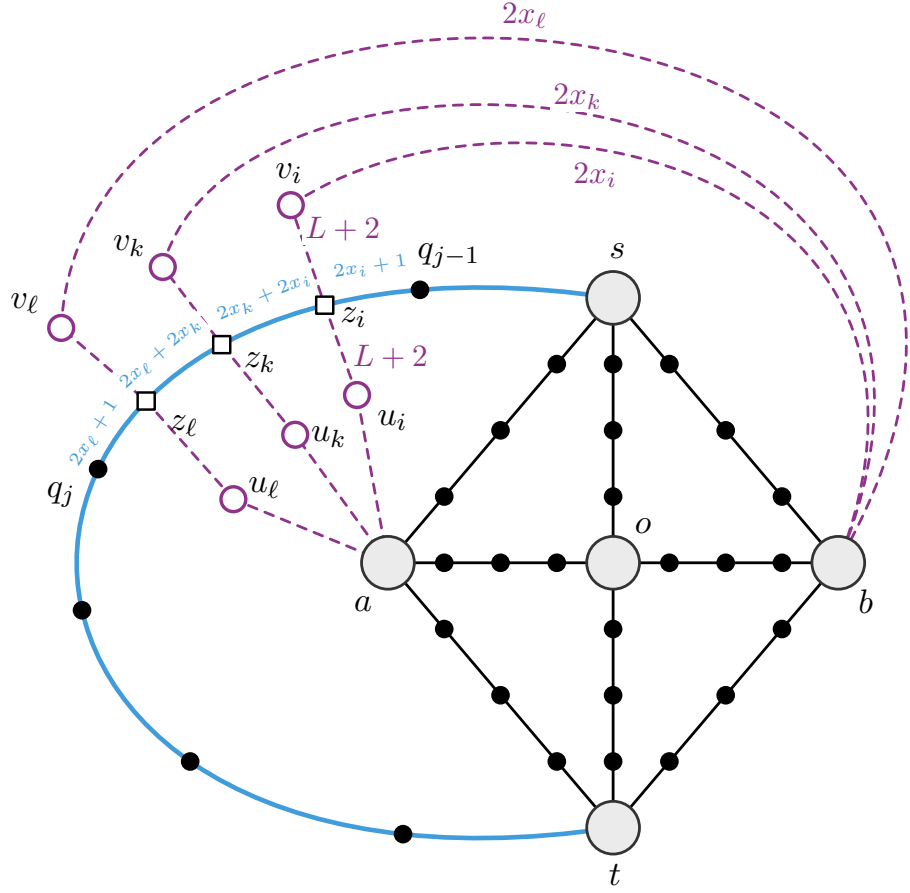
\begin{figure}[!htb]
  \centering
  \begingroup
  \definecolor{PMRblue}{RGB}{65,156,220}
  \definecolor{PMRpurple}{RGB}{143,50,139}
  \definecolor{PMRgray}{RGB}{235,235,235}

  \resizebox{.75\linewidth}{!}{%
  \begin{tikzpicture}[
    x=1cm,y=1cm,font=\normalsize,
    line cap=round,line join=round,
    frame edge/.style={
      draw=black,line width=.9pt
    },
    major vertex/.style={
      circle,draw=black!80,fill=PMRgray,
      line width=.9pt,minimum size=17pt,inner sep=0pt
    },
    terminal/.style={
      circle,fill=black,draw=none,
      minimum size=6pt,inner sep=0pt
    },
    item vertex/.style={
      circle,draw=PMRpurple,fill=white,
      line width=1.1pt,minimum size=8pt,inner sep=0pt
    },
    crossing vertex/.style={
      rectangle,draw=black,fill=white,
      line width=.8pt,minimum size=5.5pt,inner sep=0pt
    },
    spine/.style={
      draw=PMRblue,line width=1.5pt
    },
    item edge/.style={
      draw=PMRpurple,line width=.9pt,
      dash pattern=on 3pt off 2.7pt
    },
    vertex label/.style={
      inner sep=1pt
    },
    length label/.style={
      fill=white,inner sep=1pt,font=\small
    },
    spine label/.style={
      PMRblue,length label,font=\scriptsize
    }
  ]


  \coordinate (a) at (-2.55,0);
  \coordinate (b) at ( 2.55,0);
  \coordinate (s) at (0, 3);
  \coordinate (t) at (0,-3);
  \coordinate (o) at (0, 0);

  \coordinate (ui) at (-2.90,1.90);
  \coordinate (vi) at (-3.65,4.05);

  \coordinate (uk) at (-3.60,1.45);
  \coordinate (vk) at (-5.10,3.35);

  \coordinate (ul) at (-4.30,.72);
  \coordinate (vl) at (-6.25,2.66);

  \path
    (s)
    .. controls (-8.1,4.2) and (-8.1,-4.2) ..
    (t)
    coordinate[pos=.10] (qprev)
    coordinate[pos=.16] (zi)
    coordinate[pos=.24] (zk)
    coordinate[pos=.32] (zl)
    coordinate[pos=.40] (qj);


  \foreach \u/\v in {a/s,s/b,b/t,t/a,o/a,o/b,o/s,o/t}{
    \draw[frame edge] (\u)--(\v);

    \foreach \r in {.25,.50,.75}{
      \node[terminal]
        at ($(\u)!\r!(\v)$) {};
    }
  }

  \draw[spine]
    (s)
    .. controls (-8.1,4.2) and (-8.1,-4.2) ..
    (t);

  \foreach \r in {.55,.73,.89}{
    \path
      (s)
      .. controls (-8.1,4.2) and (-8.1,-4.2) ..
      (t)
      node[pos=\r,terminal] {};
  }


  \foreach \u/\z/\v in {ui/zi/vi,uk/zk/vk,ul/zl/vl}{
    \draw[item edge] (a)--(\u);
    \draw[item edge] (\u)--(\z)--(\v);
  }

  \draw[item edge]
    (vi)
    .. controls (-1.2,5.7) and (4.3,4.5) ..
    (b);

  \draw[item edge]
    (vk)
    .. controls (-4.3,6.4) and (5.0,5.9) ..
    (b);

  \draw[item edge]
    (vl)
    .. controls (-6.0,7.8) and (6.5,6.9) ..
    (b);


  \foreach \v in {a,b,s,t,o}{
    \node[major vertex] at (\v) {};
  }

  \foreach \v in {qprev,qj}{
    \node[terminal] at (\v) {};
  }

  \foreach \v in {ui,vi,uk,vk,ul,vl}{
    \node[item vertex] at (\v) {};
  }

  \foreach \v in {zi,zk,zl}{
    \node[crossing vertex] at (\v) {};
  }

  %

  \node[vertex label] at (-2.83,-.46) {$a$};
  \node[vertex label] at ( 2.86,-.43) {$b$};
  \node[vertex label] at ( .05, 3.53) {$s$};
  \node[vertex label] at ( .02,-3.53) {$t$};
  \node[vertex label] at ( .34, .42) {$o$};

  \node[vertex label]
    at (-1.86,3.48) {$q_{j-1}$};

  \node[vertex label]
    at (-6.25,.77) {$q_j$};

  \node[vertex label] at (-2.50,1.61) {$u_i$};
  \node[vertex label] at (-3.67,4.41) {$v_i$};

  \node[vertex label] at (-3.22,1.43) {$u_k$};
  \node[vertex label] at (-5.48,3.59) {$v_k$};

  \node[vertex label] at (-3.95,.85) {$u_\ell$};
  \node[vertex label] at (-6.66,2.94) {$v_\ell$};

  \node[vertex label] at (-2.94,2.78) {$z_i$};
  \node[vertex label] at (-4.00,2.27) {$z_k$};
  \node[vertex label] at (-4.86,1.55) {$z_\ell$};

  \node[spine label,rotate=9]
    at (-2.75,3.35) {\tiny{$2x_i+1$}};

  \node[spine label,rotate=23]
    at (-3.95,3.01) {\tiny{$2x_k+2x_i$}};

  \node[spine label,rotate=41]
    at (-5.12,2.38) {\tiny{$2x_\ell+2x_k$}};

  \node[spine label,rotate=58]
    at (-5.87,1.47) {\tiny{$2x_\ell+1$}};

  \node[PMRpurple,length label]
    at (-.20,4.39) {$2x_i$};

  \node[PMRpurple,length label]
    at (-.40,5.28) {$2x_k$};

  \node[PMRpurple,length label]
    at (-1.00,6.20) {$2x_\ell$};

  \node[PMRpurple,length label]
    at (-2.52,2.33) {$L+2$};

  \node[PMRpurple,length label]
    at (-3.06,3.78) {$L+2$};

  \pgfresetboundingbox
  \path[use as bounding box]
    (-7,-3.8) rectangle (3.6,6.6);

  \end{tikzpicture}%
  }
  \endgroup

  \caption{%
    Planarizing the three item paths for
    $T_j=\{x_i,x_k,x_\ell\}$. The white squares $z_i,z_k,z_\ell$, the crossing vertices, occur in this order
    on $Q[q_{j-1},q_j]$; the blue labels give the four
    consecutive spine distances, whose sum is $4B+2$. For $r\in\{i,k,\ell\}$, the purple path is $(a,u_r,z_r,v_r,b)$: its two outer edges have length
    $2x_r$, and its two middle edges have length $L+2$.
  }
  \label{fig:planarizing-triple}
\end{figure}

Suppose that the $3m$ items can be partitioned into a set  $\cT = \{T_1,\ldots, T_m\}$ of $m$ triples, each summing to $B$. We construct a planarization $G$ of $H$ as follows:
\begin{itemize}
    \item First, $G$ contains the planar drawing of $J$ as shown in \Cref{fig:frame-spine-items}(b).
    \item Next, we draw item paths. See \Cref{fig:planarizing-triple}. For each $j\in [m]$, we draw the three item paths for the three items in each triple $T_j$ that cross the spine edge $(q_{j-1},q_{j})$. Specifically, let $T_{j} = \{x_i,x_k,x_{\ell}\}$ for some $1\leq i < k < \ell \leq 3m$. We pick three vertices $z_i,z_k,z_{\ell}$ in that order along the spine segment $Q[q_{j-1}, q_{j}]$ such that:  
    \begin{equation}\label{eq:zike-pick}
        \begin{split}
            d_{Q}(z_i,q_{j-1}) &= 2x_i+1\\
        d_{Q}(z_k,z_{i}) &= 2x_i + 2x_k\\
        d_{Q}(z_{k},z_{\ell}) &= 2x_k + 2x_{\ell}\\
        d_{Q}(z_{\ell},q_{j}) &= 2x_{\ell}+1
        \end{split}
    \end{equation}
    This is possible since $x_i + x_k + x_{\ell} = B$ and $w(Q[q_{j-1}, q_{j}]) = 4B+2$. Then we planarize the three item paths  $(a, u_i, v_i,b)$, $(a, u_k, v_k,b)$, and $(a, u_\ell, v_\ell,b)$ by drawing them through $z_{i}, z_{k}, z_{\ell}$, respectively. Specifically, we draw $(a,u_i)$ inside the face of $J$ that contains the spine and $a$ on the boundary, then draw $(u_i,v_i)$ crossing $Q[q_{j-1},q_{j}]$ at $z_i$ and finally draw $(v_i,b)$ on the infinite face of $J$. The weights of the two edges $(u_i,z_i), (z_i,v_i)$ are the same, and set to be $L+2$, which sums to the weight of the edge $(u_i, v_i)$. We call $\{z_1,\ldots, z_{3m}\}$ the \EMPH{crossing vertices}.
\end{itemize}

Finally, we argue that the planarization step does not change the distances between vertices of $H$.  The proof is mostly case analysis, relying on facts from \Cref{obs:distances-H}.  

\begin{restatable}{lemma}{DistancePreservation}\label{lm:dist-preserve}
For every $x,y\in V(H)$, $d_G(x,y) = d_H(x,y)$.
\end{restatable}
\begin{proof}
   Observe that $d_G(x,y) \leq d_H(x,y)$, and that only the planarization step can potentially shortcut distances in $H$. Note that the shortest paths between vertices of the frame $F$ and vertices in $V(H)\setminus V(F)$ must still pass through one of $a,b,s,t$, and therefore, by the triangle inequality, it suffices to focus our argument on pairs that are not in $(V(F)\setminus\{a,b,s,t\})\times (V(H)\setminus V(F))$, which includes pairs of vertices within the frame $F$, and within the spine and item paths.

    The most representative case is when either $x$ or $y$ is $u_i$ for some $i\in [3m]$.  For $(x,y)$ in the set $\{(u_i,a), (u_i,u_k), (u_i,v_i)\}$, the crossing vertices do not help shorten the distances. Three types of distances that can be shortcut are $(u_i,v_k)$ for some $k\in [3m]$, or $(u_i,q_j)$ for some $j\in [0,m]$ (including the case $q_0 = s$ and $q_m = t$), or $(u_i,b)$:
    \begin{itemize}
        \item \textbf{For  $(u_i,v_k)$} with $i\not= k$, there is a new path from $u_i$ to $v_k$ going through crossing vertices, namely $$u_i\rightarrow z_i\rightarrow Q[z_i,z_k]\rightarrow v_k.$$ 
   However, since $Q[z_i,z_k] \geq 2x_i + 2x_k$ by the choice of the crossing vertices, this path has length at least $2L+4 + 2x_i+2x_k =d_H(u_i,v_k)$ in \Cref{eq:u_i-v_k-dist}. Thus, such paths do not give distances shorter than $d_H(u_i,v_k)$.
   
   \item  \textbf{For  $(u_i,q_j)$,}  there is a new shortest path from $u_i$ to $q_j$ going through $z_i$. This path has length at least $L+2 + d_G(z_i, q_j)\geq L+2+2x_i + 1 = L+3 + 2x_i$ by the choice of the crossing points. But $d_H(u_i,q_j) =  L+3 + 2x_i$ by \Cref{eq:u_i-q_j-dist}, so the distance does not change. 
   \item  \textbf{For  $(u_i,b)$,} the new shortest path must go through some spine vertex $q_j$ via $z_i$, but as shown in the previous case, $d_G(u_i,q_j) = d_H(u_i,q_j)$ and so the new shortest path will not be shorter.
    \end{itemize}

    Due to the symmetry between $\{u_i\}_{i\in [3m]}$ and $\{v_i\}_{i\in [3m]}$, the same argument excludes all pairs that have one endpoint in $\{v_i\}_{i\in [3m]}$.
    
    Now if $\pi_G(x,y)$ contains a vertex $v\in \{u_i,v_i\}_{i\in [3m]}$, then:
    \begin{equation*}
        d_G(x,y) = d_G(x,v) + d_G(v,y) = d_H(x,v)+d_H(v,y) \geq d_H(x,y)~.
    \end{equation*}
   Since $d_G(x,y)\leq d_H(x,y)$, it must be that $d_G(x,y) = d_H(x,y)$.
    
    If $\pi_G(x,y)$ contains no such vertex, then it lies entirely in the subdivision of $J$ obtained by inserting the crossing vertices along the spine. Contracting the crossing vertices gives an $x$-to-$y$ path in  $J$ of the same length.
\end{proof}

\subsection{Soundness}\label{subsec:soundness}

Here we are given a planar graph $\widehat{G}$ with edge-weight function $\hat{w}$ that realizes $(V(H),d_H)$. We call vertices of $H$ in $\widehat{G}$ \EMPH{terminals}.  We assume that $\widehat{G}$ contains a minimal number of vertices, and specifically, $\widehat{G}$ does not contain any non-terminal vertex of degree 2 since we can always contract such a vertex. Our goal is to extract a solution to the \ThreePartition~ problem from $\widehat{G}$. At a very high level, the strategy is to show that $\widehat{G}$ is ``close'' to a planarization of $H$. 

Let  $\EMPH{\text{$\widehat{F}$}} = \cup_{(u,v)\in E(F)}\pi_{\widehat{G}}(u,v)$ be the copy of the frame $F$ in $\widehat{G}$. Let $\EMPH{\text{$\widehat{Q}$}} = \cup_{(u,v)\in E(Q)}\pi_{\widehat{G}}(u,v)$ be the copy of the spine $Q$ in $\widehat{G}$. For every $i\in [3m]$, let $\EMPH{\text{$\widehat{P}_i$}}= \pi_{\widehat{G}}(a,u_i)\circ \pi_{\widehat{G}}(u_i,v_i)\circ \pi_{\hat{G}}(v_i,b)$, which is the copy of the item path $(a,u_i,v_i,b)$.  The following lemma formalizes the intuition that $\widehat{G}$ is close to a planarization of $H$. Specifically:

\begin{lemma}\label{lm:planarization-similar} The following claims hold:
\begin{enumerate}
    \item \textbf{Frame does not contain extra terminals:} $V(\widehat{F})\cap V(H)\setminus V(F) = \varnothing$. 
    \item  \textbf{Paths of disjoint frame edges are (internally) disjoint:} For any two distinct edges $(a,b)$ and $(x,y)$ in $E(F)$, $\pi_{\widehat{G}}(a,b)\cap \pi_{\widehat{G}}(x,y) = \{a,b\}\cap \{x,y\}$.
    \item \textbf{Spine is embedded outside the frame:}  $V(\widehat{Q})\cap V(\widehat{F}) = \{s,t\}$.
    \item \textbf{Item paths avoid the frame:} for every $i\in [3m]$,  $\widehat{P}_i \cap V(\widehat{F}) = \{a,b\}$.    
    \item  \textbf{Item paths are simple:}  for every $i\in [3m]$, $\widehat{P}_i$ is a simple path.
\end{enumerate}
\end{lemma}

 A key idea in the proof is the following anti-Monge inequality: it forces two paths to be disjoint if their terminal distances satisfy one of the strict inequalities.

\begin{wrapfigure}{r}{0.5\textwidth}
  \centering
  \begingroup
  \definecolor{PMRblue}{RGB}{34,102,138}
  \definecolor{PMRorange}{RGB}{181,82,45}
  \definecolor{PMRdark}{RGB}{39,55,67}

  \resizebox{\linewidth}{!}{%
  \begin{tikzpicture}[
    x=1cm,y=1cm,font=\normalsize,
    line cap=round,line join=round,
    terminal/.style={
      circle,draw=none,fill=PMRdark,
      minimum size=5.5pt,inner sep=0pt
    },
    intersection/.style={
      circle,draw=PMRdark,fill=white,
      line width=.9pt,minimum size=6.5pt,inner sep=0pt
    },
    distance label/.style={
      fill=white,inner sep=1.5pt
    }
  ]

    \coordinate (x) at (0, .95);
    \coordinate (u) at (0,-.95);
    \coordinate (v) at (6.8, .95);
    \coordinate (y) at (6.8,-.95);
    \coordinate (z) at (3.4,0);

    \draw[PMRblue,line width=1.2pt]
      (x)
      .. controls (2.6,.95) and (4.2,-.95) ..
      (y);

    \draw[PMRorange,line width=1.2pt]
      (u)
      .. controls (2.6,-.95) and (4.2,.95) ..
      (v);

    \foreach \p in {x,u,v,y}{
      \node[terminal] at (\p) {};
    }
    \node[intersection] at (z) {};

    \node at (-.32, .97) {$x$};
    \node at (-.32,-.95) {$u$};
    \node at (7.12, .97) {$v$};
    \node at (7.12,-.95) {$y$};
    \node at (3.50, .36) {$z$};

    \node[PMRblue,distance label]
      at (1.45,.91) {$d_{\widehat{G}}(x,z)$};

    \node[PMRblue,distance label]
      at (5.35,-.91) {$d_{\widehat{G}}(y,z)$};

    \node[PMRorange,distance label]
      at (1.45,-.91) {$d_{\widehat{G}}(u,z)$};

    \node[PMRorange,distance label]
      at (5.35,.91) {$d_{\widehat{G}}(v,z)$};

  \end{tikzpicture}%
  }
  \endgroup

  \caption{%
    $\pi_{\widehat{G}}(x,y)$ and $\pi_{\widehat{G}}(u,v)$
    intersect at a vertex $z$.
  }
  \label{fig:path-switching}
\end{wrapfigure}
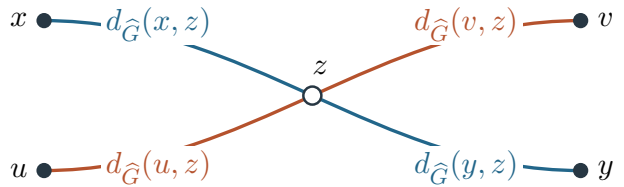

\begin{lemma}[Anti-Monge Inequality]\label{lm:anti-monge} Let $x,y,u,v$ be four vertices in  $\widehat{G}$. 
If \EMPH{at least one} of the following inequalities is satisfied 
\begin{equation}
    \begin{split}
        d_{\widehat{G}}(x,u) + d_{\widehat{G}}(v,y) &>    d_{\widehat{G}}(x,y) +   d_{\widehat{G}}(u,v)\\
     d_{\widehat{G}}(x,v) + d_{\widehat{G}}(u,y) &>    d_{\widehat{G}}(x,y) +   d_{\widehat{G}}(u,v)~,
    \end{split}
\end{equation}
then $\pi_{\widehat{G}}(x,y)$ and $\pi_{\widehat{G}}(u,v)$ are vertex-disjoint.
\end{lemma}
\begin{proof}
Suppose otherwise, let $z \in\pi_{\widehat{G}}(x,y)\cap \pi_{\widehat{G}}(u,v)$ be a shared vertex; see \Cref{fig:path-switching}.  Then by the triangle inequality:
\begin{equation*}
    \begin{split}
        d_{\widehat{G}}(x,u) + d_{\widehat{G}}(v,y)  &\leq  d_{\widehat{G}}(x,z) +  d_{\widehat{G}}(z,u)  + d_{\widehat{G}}(v,z) + d_{\widehat{G}}(z,y)\\
        &= (d_{\widehat{G}}(x,z) +  d_{\widehat{G}}(z,y)) + (d_{\widehat{G}}(u,z) +  d_{\widehat{G}}(z,v))\\
        &=  d_{\widehat{G}}(x,y) +   d_{\widehat{G}}(u,v)
    \end{split}
\end{equation*}

By the same argument, $d_{\widehat{G}}(x,v) + d_{\widehat{G}}(u,y) \leq  d_{\widehat{G}}(x,y) +   d_{\widehat{G}}(u,v)$. Then the lemma follows by contraposition.
\end{proof}

Now proving  \Cref{lm:planarization-similar} comes down to showing that paths in $\widehat{G}$ realizing edges of $H$ are vertex disjoint, and we can simply check whether their distances satisfy at least one of the inequalities in \Cref{lm:anti-monge}. The proof is somewhat mechanical, and we delay it to the end, in \Cref{subsec:planrization-proof}.

The first two items of \Cref{lm:planarization-similar}, together with the fact that $\widehat{G}$ is minimal, imply that  $\widehat{F}$ is a subdivision of $F$. 

Note that $F$ admits a unique embedding on the sphere---if we contract all degree-2 vertices in $F$, we obtain a 3-connected graph on the vertex set $\{s,a,b,t,o\}$ which therefore has a unique embedding by Whitney's embedding theorem. Thus, $\widehat{F}$ also has a unique embedding on the sphere. By items 3 and 4 in \Cref{lm:planarization-similar}, the spine copy $\widehat{Q}$ and the item paths $\{\widehat{P}_i\}_{i\in [3m]}$ are embedded ``outside'' the frame $\widehat{F}$. 

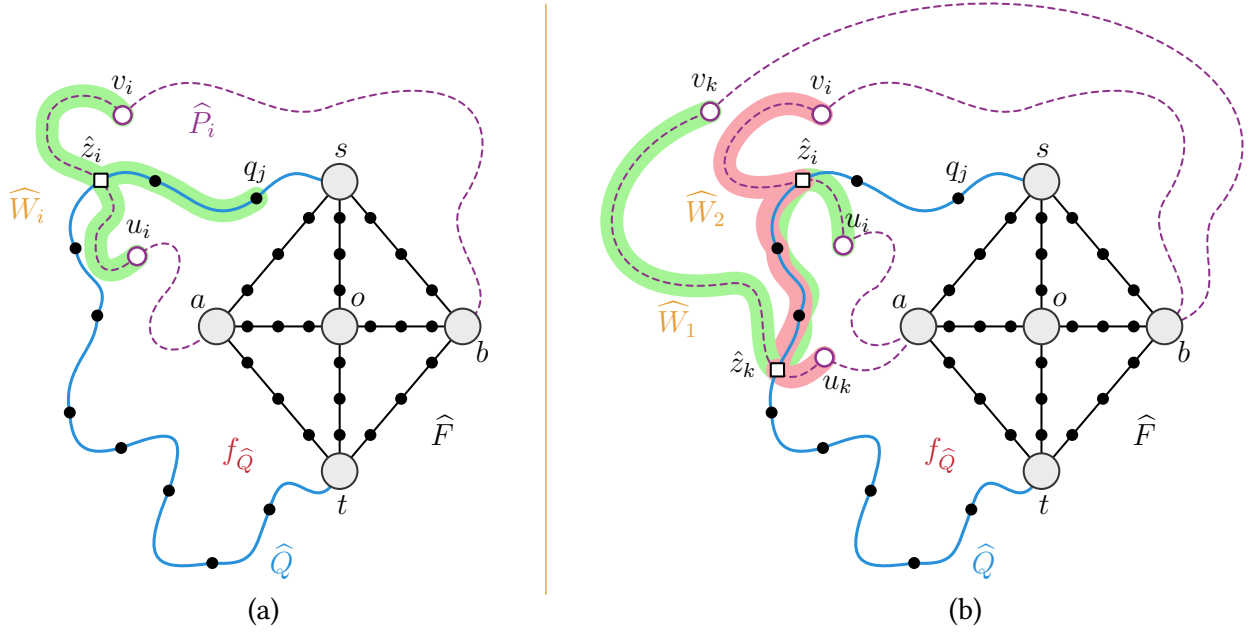
\begin{figure}[!htb]
\centering
\begingroup
\definecolor{PMRblue}{RGB}{40,145,218}
\definecolor{PMRpurple}{RGB}{143,50,139}
\definecolor{PMRorange}{RGB}{230,158,53}
\definecolor{PMRgray}{RGB}{235,235,235}
\definecolor{PMRgreen}{RGB}{165,244,149}
\definecolor{PMRpink}{RGB}{249,166,175}
\definecolor{PMRred}{RGB}{195,35,45}

\resizebox{\linewidth}{!}{%
\begin{tikzpicture}[
  x=1cm,y=1cm,font=\large,
  line cap=round,line join=round,
  frame edge/.style={
    draw=black,line width=.8pt
  },
  major vertex/.style={
    circle,draw=black!80,fill=PMRgray,
    line width=.8pt,minimum size=14pt,inner sep=0pt
  },
  terminal/.style={
    circle,fill=black,draw=none,
    minimum size=4.7pt,inner sep=0pt
  },
  item vertex/.style={
    circle,draw=PMRpurple,fill=white,
    line width=1pt,minimum size=7pt,inner sep=0pt
  },
  crossing vertex/.style={
    rectangle,draw=black,fill=white,
    line width=.8pt,minimum size=5.2pt,inner sep=0pt
  },
  spine/.style={
    draw=PMRblue,line width=1.2pt
  },
  item path/.style={
    draw=PMRpurple,line width=.8pt,
    dash pattern=on 2.5pt off 2pt
  },
  green walk/.style={
    draw=PMRgreen,line width=9pt
  },
  pink walk/.style={
    draw=PMRpink,line width=9pt
  },
  vertex label/.style={
    inner sep=1pt
  }
]


\newcommand{\PMRsetupWalks}{%
  \coordinate (a) at (-1.7,0);
  \coordinate (b) at ( 1.7,0);
  \coordinate (s) at (0, 2);
  \coordinate (t) at (0,-2);
  \coordinate (o) at (0,0);

  \coordinate (zi) at (-3.30,2.02);
  \coordinate (zk) at (-3.65,-.60);

  \coordinate (qtop) at (-2.55,2.01);
  \coordinate (qj) at (-1.15,1.77);
  \coordinate (qmone) at (-3.65,1.08);
  \coordinate (qmtwo) at (-3.35,.15);

  \coordinate (qlone) at (-3.73,-1.19);
  \coordinate (qltwo) at (-3.02,-1.68);
  \coordinate (qlthree) at (-2.36,-2.27);
  \coordinate (qlfour) at (-1.76,-3.27);
  \coordinate (qlfive) at (-.97,-2.53);
}


\def\PMRziQj{%
  (zi)
  .. controls (-3.04,2.18) and (-2.83,2.15) .. (qtop)
  .. controls (-2.05,1.85) and (-1.78,1.33) .. (qj)
}

\def\PMRmiddleQ{%
  (zi)
  .. controls (-3.70,1.75) and (-3.82,1.39) .. (qmone)
  .. controls (-3.56,.57) and (-3.11,.67) .. (qmtwo)
  .. controls (-3.35,-.17) and (-3.55,-.36) .. (zk)
}

\def\PMRlowerQ{%
  (zk)
  .. controls (-3.75,-.80) and (-3.77,-1.02) .. (qlone)
  .. controls (-3.70,-1.78) and (-3.39,-1.77) .. (qltwo)
  .. controls (-2.26,-1.39) and (-2.09,-1.43) .. (qlthree)
  .. controls (-2.80,-3.31) and (-2.51,-3.38) .. (qlfour)
  .. controls (-1.16,-3.31) and (-1.19,-3.11) .. (qlfive)
  .. controls (-.71,-1.62) and (-.34,-2.81) .. (t)
}

\def\PMRviB{%
  (vi)
  .. controls (-1.90,4.05) and (-1.09,2.62) .. (-.10,3.05)
  .. controls (1.92,3.70) and (2.30,2.69) .. (1.74,1.81)
  .. controls (1.27,1.16) and (2.43,.50) .. (b)
}


\newcommand{\PMRdrawBase}{%
  \foreach \u/\v in {a/s,s/b,b/t,t/a,o/a,o/b,o/s,o/t}{
    \draw[frame edge] (\u)--(\v);

    \foreach \r in {.25,.50,.75}{
      \node[terminal] at ($(\u)!\r!(\v)$) {};
    }
  }

  \draw[spine] \PMRziQj;

  \draw[spine]
    (qj)
    .. controls (-.53,2.40) and (-.48,1.96) ..
    (s);

  \draw[spine] \PMRmiddleQ;
  \draw[spine] \PMRlowerQ;
}

\newcommand{\PMRbaseVertices}{%
  \foreach \v in {qtop,qj,qmone,qmtwo,qlone,qltwo,qlthree,qlfour,qlfive}{
    \node[terminal] at (\v) {};
  }

  \foreach \v in {a,b,s,t}{
    \node[major vertex] at (\v) {};
  }

  \node[major vertex,fill=PMRgray] at (o) {};

  \node[vertex label] at (-1.96,0.35) {$a$};
  \node[vertex label] at ( 1.97,-.35) {$b$};
  \node[vertex label] at ( .02, 2.43) {$s$};
  \node[vertex label] at ( .04,-2.46) {$t$};
  \node[vertex label] at ( .25, .38) {$o$};
  \node[vertex label] at (-1.16,2.13) {$q_j$};

  \node at (1.42,-1.43) {$\widehat{F}$};

  \node[PMRblue]
    at (-.80,-3.25) {$\widehat{Q}$};

  \node[PMRred]
    at (-1.40,-1.78) {$f_{\widehat{Q}}$};
}

\draw[PMRorange!80,line width=.65pt]
  (2.85,-3.70)--(2.85,4.45);


\begin{scope}
  \PMRsetupWalks

  \coordinate (ui) at (-2.80,.97);
  \coordinate (vi) at (-3.00,2.92);

  \def\PMRuiZi{%
    (ui)
    .. controls (-3.20,.61) and (-3.43,.78) .. (-3.36,1.27)
    .. controls (-3.30,1.52) and (-2.90,1.48) .. (zi)
  }

  \def\PMRziVi{%
    (zi)
    .. controls (-3.60,2.27) and (-4.10,2.17) .. (-4.04,2.75)
    .. controls (-4.08,3.24) and (-3.30,3.35) .. (vi)
  }

  \draw[green walk] \PMRuiZi;
  \draw[green walk] \PMRziVi;
  \draw[green walk] \PMRziQj;

  \PMRdrawBase

  \draw[item path]
    (a)
    .. controls (-2.40,-.70) and (-2.95,-.12) .. (-2.40,.47)
    .. controls (-2.01,.92) and (-2.27,1.42) .. (ui);

  \draw[item path] \PMRuiZi;
  \draw[item path] \PMRziVi;
  \draw[item path] \PMRviB;

  \PMRbaseVertices

  \node[item vertex] at (ui) {};
  \node[item vertex] at (vi) {};
  \node[crossing vertex] at (zi) {};

  \node[vertex label]
    at (-2.78,1.35) {$u_i$};

  \node[vertex label]
    at (-3.00,3.34) {$v_i$};

  \node[vertex label]
    at (-3.43,2.45) {$\hat z_i$};

  \node[PMRpurple]
    at (-1.90,2.85) {$\widehat{P}_i$};

  \node[PMRorange]
    at (-4.32,1.67) {$\widehat{W}_i$};

  \node at (-1.05,-3.94) {(a)};
\end{scope}


\begin{scope}[xshift=9.7cm]
  \PMRsetupWalks

  \coordinate (ui) at (-2.74,1.12);
  \coordinate (vi) at (-3.04,2.93);
  \coordinate (uk) at (-3.00,-.43);
  \coordinate (vk) at (-4.58,2.97);

  \def\PMRuiZi{%
    (ui)
    .. controls (-2.72,1.66) and (-2.88,2.09) ..
    (zi)
  }

  \def\PMRziVi{%
    (zi)
    .. controls (-4.47,1.68) and (-4.61,2.28) .. (-4.10,2.79)
    .. controls (-3.69,3.18) and (-3.35,3.12) .. (vi)
  }

  \def\PMRukZk{%
    (uk)
    .. controls (-3.24,-.69) and (-3.50,-.79) ..
    (zk)
  }

  \def\PMRzkVk{%
    (zk)
    .. controls (-4.06,-.38) and (-3.57,.71) .. (-4.35,.59)
    .. controls (-6.52,.29) and (-6.34,2.37) .. (vk)
  }

  \begin{scope}[xshift=2.8pt]
    \draw[green walk] \PMRmiddleQ;
  \end{scope}

  \begin{scope}[xshift=-2.8pt]
    \draw[pink walk] \PMRmiddleQ;
  \end{scope}

  \draw[green walk] \PMRuiZi;
  \draw[green walk] \PMRzkVk;

  \draw[pink walk] \PMRukZk;
  \draw[pink walk] \PMRziVi;

  \PMRdrawBase

  \draw[item path]
    (a)
    .. controls (-2.55,-.71) and (-2.97,.23) .. (-2.28,.54)
    .. controls (-1.75,1.12) and (-2.15,1.57) .. (ui);

  \draw[item path]
    (a)
    .. controls (-2.08,-.67) and (-2.44,-.74) ..
    (uk);

  \draw[item path] \PMRuiZi;
  \draw[item path] \PMRziVi;
  \draw[item path] \PMRukZk;
  \draw[item path] \PMRzkVk;

  \draw[item path] \PMRviB;

  \draw[item path]
    (vk)
    .. controls (-1.46,6.01) and (4.35,3.83) .. (2.38,1.56)
    .. controls (2.07,.97) and (2.76,.31) .. (b);

  \PMRbaseVertices

  \foreach \v in {ui,vi,uk,vk}{
    \node[item vertex] at (\v) {};
  }

  \foreach \v in {zi,zk}{
    \node[crossing vertex] at (\v) {};
  }

  \node[vertex label]
    at (-2.54,1.43) {$u_i$};

  \node[vertex label]
    at (-3.04,3.34) {$v_i$};

  \node[vertex label]
    at (-2.82,-.80) {$u_k$};

  \node[vertex label]
    at (-4.65,3.37) {$v_k$};

  \node[vertex label]
    at (-3.23,2.43) {$\hat z_i$};

  \node[vertex label]
    at (-4.13,-.52) {$\hat z_k$};

  \node[PMRorange]
    at (-5.03,.10) {$\widehat{W}_1$};

  \node[PMRorange]
    at (-4.64,1.63) {$\widehat{W}_2$};

  \node at (-1.05,-3.94) {(b)};
\end{scope}

\pgfresetboundingbox
\path[use as bounding box]
  (-4.75,-4.20) rectangle (12.50,4.65);

\end{tikzpicture}%
}
\endgroup

\caption{%
  Walks used to bound distances along $\widehat{Q}$.
  (a) The green highlight shows $\widehat{W}_i$:
  $u_i\to\hat z_i\to q_j\to\hat z_i\to v_i$;
  the spine segment from $\hat z_i$ to $q_j$ is traversed twice.
  (b) Green shows $\widehat{W}_1$ from $u_i$ through
  $\hat z_i,\hat z_k$ to $v_k$; pink shows $\widehat{W}_2$
  from $u_k$ through $\hat z_k,\hat z_i$ to $v_i$. Purple dashed curves depict the item paths.%
}
\label{fig:intersection-walks}
\end{figure}

Let \EMPH{$f_{\widehat{Q}}$} be the face in the embedding of $\widehat{Q}\cup \widehat{F}$ inherited from $\widehat{G}$ containing $a$ and $\widehat{Q}$ on the boundary; see \Cref{fig:intersection-walks}(a). On the plane, by designating the face containing $b$ and $\widehat{Q}$ on the boundary as the infinite face of $\widehat{Q}\cup \widehat{F}$, $f_{\widehat{Q}}$ is finite. Next, we show exactly how each item path $\widehat{P}_i$ embeds on the plane with respect to $\widehat{Q}\cup \widehat{F}$. Specifically, we show that the intersection of $\widehat{P}_i$ and $\widehat{Q}$ must be between $u_i$ and $v_i$, which also means that $u_i$ is embedded strictly inside $f_{\widehat{Q}}$ while $v_i$ is embedded outside.

\begin{lemma}\label{lm:itempath-spine-intersect} For every $i\in [3m]$, $\varnothing\not= \widehat{P}_i\cap \widehat{Q} \subseteq \widehat{P}_i[u_i,v_i]$.
\end{lemma}
\begin{proof}
    Observe that $d_H(s,u_i) = d_H(s,a) + d_H(a,u_i) = L+3 + 2x_i$, and $d_H(t,a) = L+3$. Observe that:
    \begin{equation*}
        \begin{split}
             d_{\widehat{G}}(s,u_i) +  d_{\widehat{G}}(t,a)  &= L+3 + 2x_i + L+3 =  2L + 6 +2x_i\\
             &> L + 2x_i = d_H(s,t) + d_H(a,u_i)\\
             &= d_{\widehat{G}}(s,t) + d_{\widehat{G}}(a,u_i)
        \end{split}
    \end{equation*}
    Thus, $\widehat{Q} = \pi_{\widehat{G}}(s,t)$ and $\widehat{P}_i[a,u_i] = \pi_{\widehat{G}}(a,u_i)$ are vertex disjoint by \Cref{lm:anti-monge}. By the same argument, $\hat{Q}$ and $\widehat{P}_i[b,v_i]$ are vertex disjoint. Therefore, $\widehat{P}_i\cap \widehat{Q} \subseteq \widehat{P}_i[u_i,v_i]$.  We note that $a$ and $b$ are separated by the Jordan curve $\widehat{Q}\circ \pi_{\widehat{F}}(t,o)\circ \pi_{\widehat{F}}(o,s)$ and hence  $\widehat{P}_i\cap \widehat{Q} \not=\varnothing$.
\end{proof}

The next lemma shows that the intersections of the item paths and the spine $\{\widehat{P}_i\cap \widehat{Q}\}$ are far from each other, and far from $V(Q)$.

\begin{lemma}\label{lm:intersection-far} For every $i\in [3m]$, pick (an arbitrary) $\hat{z}_i \in \widehat{P}_i\cap \widehat{Q}$. Then:
\begin{align}
    \hat{w}(\widehat{Q}[\hat{z}_i,q_j]) &\geq 2x_i + 1  &\qquad \forall j\in [0,m] \label{eq:qj-zi-dist}\\
    \hat{w}(\widehat{Q}[\hat{z}_i, \hat{z}_k])  &\geq 2x_i + 2x_k  &\qquad \forall i\not=k \in [3m]  \label{eq:zi-zk-dist}
\end{align}
\end{lemma}
\begin{proof}
    Let $\widehat{W}_i$ be the walk from $u_i$ to $q_j$ and $q_j$ to $v_i$ going through $\hat{z}_i$; see \Cref{fig:intersection-walks}(a). Then we have:
    \begin{equation*}
       \begin{split}
            \hat{w}(\widehat{W}_i) &\geq  d_H(u_i,q_j) +  d_H(q_j,v_i) \\
        &= 2(L+3+2x_i) \qquad \text{(by \Cref{eq:u_i-q_j-dist})}
       \end{split}
    \end{equation*}
However, $ \hat{w}(\widehat{W}_i)  =  d_H(u_i,v_i)  + 2\hat{w}(\widehat{Q}[\hat{z}_i,q_j])$, giving:
\begin{equation*}
    \begin{split}
    \hat{w}(\widehat{Q}[\hat{z}_i,q_j]) &\geq (2(L+3+2x_i) -  d_H(u_i,v_i))/2 \\
        &\geq (2(L+3+2x_i) -  (2L+4))/2 \qquad \text{(by \Cref{eq:ui-vi-dist})}\\
        &= 2x_i+1
    \end{split}
\end{equation*}
as claimed. 

For the second inequality, we take the walk $\widehat{W}_{1}$ from $u_i$ to $v_k$ going from $u_i$ to $\hat{z}_i$, then to $\hat{z}_k$ (along $\widehat{Q}$), and finally to $v_k$, and the walk $\widehat{W}_{2}$ from $u_k$ to $v_i$ going from $u_k$ to $\hat{z}_k$, then to $\hat{z}_i$ (along $\widehat{Q}$), and finally to $v_i$. See \Cref{fig:intersection-walks}(b). Observe that:
\begin{equation*}
   \begin{split}
        \hat{w}(\widehat{W}_1) + \hat{w}(\widehat{W}_2) &= d_H(u_i,v_i) +  d_H(u_k,v_k) + 2 \hat{w}(\widehat{Q}[\hat{z}_i, \hat{z}_k])\\
        &= 2(2L+4) +   2 \hat{w}(\widehat{Q}[\hat{z}_i, \hat{z}_k])\qquad \text{(by \Cref{eq:ui-vi-dist})}
   \end{split}
\end{equation*}
However, by definition of $\widehat{W}_1$ and $\widehat{W}_2$,
\begin{equation*}
\begin{split}
     \hat{w}(\widehat{W}_1) + \hat{w}(\widehat{W}_2) &\geq d_H(u_i,v_k) +    d_H(u_k,v_i) \\
     &= 2(2L+4+2x_i + 2x_k) \qquad \text{(by \Cref{eq:u_i-v_k-dist})}
\end{split}
\end{equation*}
Thus, we have:
\begin{equation*}
    2\hat{w}(\widehat{Q}[\hat{z}_i, \hat{z}_k]) \geq 2(2L+4+2x_i + 2x_k) - 2(2L+4) = 4x_i + 4x_k~,
\end{equation*}
giving $\hat{w}(\widehat{Q}[\hat{z}_i, \hat{z}_k])\geq 2x_i + 2x_k$ as claimed.
\end{proof}

We emphasize that \Cref{eq:qj-zi-dist} holds for every $q_j$, not just vertices closest to $\hat{z}_i$.

\paragraph{Extracting a solution to \ThreePartition:}  For every $i \in [3m]$, let $I_i$ be the closed interval of $\widehat{Q}$ centered at $\hat{z}_i$ and length $4x_i$. \Cref{eq:qj-zi-dist} in \Cref{lm:intersection-far} means that $I_i$ does not contain any spine vertex $q_j$. For every $j\in [m-1]$, let $I'_j$ be the interval of $\widehat{Q}$ centered at $q_j$ and length $2$. Let $I'_0$ and $I'_{m}$ be the prefix and suffix of $\widehat{Q}$ of length $1$; they contain $q_0 = s$ and $q_{m}=t$ as endpoints, respectively. By \Cref{lm:intersection-far}, all the intervals are internally disjoint. Thus:
\begin{equation}
   \begin{split}
        m(4B+2) &= \hat{w}(\widehat{Q})  \\
        &\geq \sum_{i\in [3m]} |I_i| + \sum_{j\in [0,m]} |I'_j|\\
        &= 4(\sum_{i}x_i)  + 2m = 4mB + 2m = m(4B+2)
   \end{split}
\end{equation}
Thus, $\widehat{Q}$ is completely covered by these intervals. 

For every $j\in [m]$, note that $\hat{w}(\widehat{Q}[q_{j-1},q_j]) = 4B+2$. Since $x_{i} < B/2$ for every $i\in [3m]$, $\widehat{Q}[q_{j-1},q_j]$ must contain at least 3 intervals in $\{I_i\}_{i\in [3m]}$, as otherwise, the total length of 2 intervals is less than $4B$. However, since there are exactly $m$ such $\widehat{Q}[q_{j-1},q_j]$, each  $\widehat{Q}[q_{j-1},q_j]$ must contain exactly 3 intervals in $\{I_i\}_{i\in [3m]}$, say $I_{i}, I_{k}, I_{\ell}$. Then it must be that $$4(x_i + x_k + x_{\ell}) + 2 = \hat{w}(\widehat{Q}[q_{j-1},q_j]) = 4B+2,$$ giving $x_i+x_k+x_{\ell} = B$, a solution of \ThreePartition.

\subsection[Proof of the planarization lemma]{Proof of \Cref{lm:planarization-similar}}\label{subsec:planrization-proof}

 Item 1 of \Cref{lm:planarization-similar} follows from items 3, 4, and 5; therefore, we focus on proving items 2-5.  As mentioned above, the key idea of the proof is the anti-Monge inequality in \Cref{lm:anti-monge}. 
 
\paragraph{2) Paths of disjoint frame edges are disjoint.}    For any two distinct edges $(u,v)$ and $(x,y)$ in $E(F)$, we want to show that $\pi_{\widehat{G}}(u,v)\cap \pi_{\widehat{G}}(x,y) = \{u,v\}\cap \{x,y\}$. Consider the case where the two edges share the same endpoint, say $x = u$. In $H$, the shortest path between $v$ and $y$ consists of only two edges $(v,u=x), (x,y)$. If $\pi_{\widehat{G}}(u,v)$
  and $\pi_{\widehat{G}}(x=u,y)$ intersect at any vertex $z\not= u$, then $\hat{w}(\pi_{\widehat{G}}(x=u,z)) > 0$ and hence $$d_{H}(v,y) = d_{\widehat{G}}(v,y) < d_{\widehat{G}}(v,u) + d_{\widehat{G}}(x,y) = d_{H}(v,u) + d_{H}(x,y)~,$$ a contradiction.

  Next, we consider the case where $(u,v)$ and $(x,y)$ do not share any endpoint.  If $d_H(u,x) + d_H(v,y) > d_H(u,v) + d_H(x,y)$, then  $d_{\widehat{G}}(u,x) + d_{\widehat{G}}(v,y) > d_{\widehat{G}}(u,v) + d_{\widehat{G}}(x,y)$, and by the anti-Monge inequality in \Cref{lm:anti-monge}, $\pi_{\widehat{G}}(u,v)$ and $\pi_{\widehat{G}}(x,y)$ are vertex-disjoint. Otherwise,  $d_H(u,x) + d_H(v,y) \leq d_H(u,v) + d_H(x,y)$. Consider $T = \pi_H(u,x)\cup \pi_H(y,v)\cup \{(u,v), (x,y)\}$ which has total edge weight at most $2(d_H(u,v) + d_H(x,y)) \leq 2(\floor{B/2} + \floor{B/2}) \leq 2B$. However, the (weighted) girth of $H$ is at least $L+2 > 2B$; thus $T$ contains no cycle; it is a tree. Let $W$ be the shortest path in $T$ connecting two edges $(u,v)$ and $(x,y)$. Note that  $w(W) > 0$ since the two edges do not share any endpoint.  Then:
  \begin{equation*}
     \begin{split}
          d_H(u,x) + d_H(v,y) &=   d_T(u,x) + d_T(v,y) \\
          &= w(u,v)+ w(x,y) + 2w(W) \\
          &>  w(u,v)+ w(x,y) = d_H(u,v) + d_H(x,y) 
     \end{split}
  \end{equation*}
  contradicting the assumption that $d_H(u,x) + d_H(v,y) \leq d_H(u,v) + d_H(x,y)$.

\paragraph {3) Spine is embedded outside the frame.} Our goal is to show that $V(\widehat{Q})\cap V(\widehat{F}) = \{s,t\}$. Let $K = F\cup \{(s,t)\}$ be the graph obtained from $F$ by adding an edge $(s,t)$ of weight exactly $L$. 
Let $(x,y)$ be any edge of $F$. If $s = x$, observe that in $K$, the shortest path from $t$ to $y$ must be $(t,s)\circ (s,y)$. Thus, if  $\widehat{Q}\cap \pi_{\widehat{G}}(s,y)\not=\{s\}$, by the same argument in proving 2), $d_{H}(t,y) = d_{\widehat{G}}(t,y) < d_{H}(t,y)$, we get a contradiction. 

Herein, we assume that $\{x,y\}\cap  \{s,t\} = \varnothing$. The argument is similar to the proof of 2) for the case where two edges share no endpoint. First, we observe that:
\begin{observation}\label{obs:K-girth} $K$ has (weighted) girth $3L+6$.
\end{observation}
\begin{proof}
    A cycle in $K$ includes either (i) edge $(s,t)$, and at least two paths between $\{a,b,s,t,o\}$, totaling $L+2(L+3) = 3L+6$, or (ii) at least three paths between $\{a,b,s,t,o\}$, totaling at least $2(L+2) +  L+3 = 3L+7$.
\end{proof}
If $d_K(s,x)+d_K(t,y)> d_K(s,t) + d_K(x,y)$, since \(K\) and \(\widehat G\) agree on frame-terminal distances, by the anti-Monge inequality, $\pi_{\widehat{G}}(s,t)\cap \pi_{\widehat{G}}(x,y)=\varnothing$ as claimed. Thus, we assume that $d_K(s,x)+d_K(t,y) \leq d_K(s,t) + d_K(x,y)$; our goal is to derive a contradiction. 

Let $T = \pi_K(s,x)\cup  \pi_K(t,y) \cup \{(s,t),(x,y)\}$. Then the total length of all edges of $T$ is at most $2(d_K(s,t) + d_K(x,y)) \leq 2(L + B/2) = 2L+B$. By \Cref{obs:K-girth}, the girth of $K$ is strictly larger than $2L+B$ and hence $T$ is a tree.  Let $W$ be the shortest path in $T$ connecting two edges  $(s,t)$ and $(x,y)$. Then:
 \begin{equation*}
     \begin{split}
          d_K(s,x) + d_K(t,y) &=   d_T(s,x) + d_T(t,y) \\
          &= d_T(s,t)+ d_T(x,y) + 2w(W) \\
          &>   d_T(s,t)+ d_T(x,y)  = d_K(s,t) + d_K(x,y) 
     \end{split}
  \end{equation*}
a contradiction.  

\paragraph{4) Item paths avoid the frame.} We want to show that for every $i\in [3m]$,  $\widehat{P}_i \cap V(\widehat{F}) = \{a,b\}$. Let $(x,y)$ be any edge of $F$.  

First, we focus on $\widehat{P}_i[a,u_i]$; the same argument applies to  $\widehat{P}_i[v_i,b]$ by the symmetry between $a$ and $b$. If $a\not\in \{x,y\}$, then $d_H(a,x) + d_H(a,y) > d_H(x,y)$. We have:
\begin{equation*}
\begin{split}
    d_H(u_i,x) + d_H(a,y) &= 2x_i + d_J(a,x) + d_H(a,y) \qquad \text{(by \Cref{eq:ui-z-inJ})}\\
    &= 2x_i + d_H(a,x) + d_H(a,y)\\
    &>   2x_i + d_H(x,y) = d_H(a,u_i) + d_H(x,y)
\end{split}
\end{equation*}
Thus, $d_{\widehat{G}}(u_i,x) + d_{\widehat{G}}(a,y) >  d_{\widehat{G}}(a,u_i) + d_{\widehat{G}}(x,y)$ and hence by the anti-Monge inequality, $\widehat{P}_i[a,u_i]\cap \pi_{\widehat{G}}(x,y) = \varnothing$. If $a\in \{x,y\}$ (w.l.o.g. $a = x$), then $d_H(u_i,y) = 2x_i + d_H(a,y)$, and by the same argument in proving 2), an intersection between $\widehat{P}_i[a,u_i]$ and $\pi_{\widehat{G}}(a,y)$ other than $a$ will result in $d_{\widehat{G}}(u_i,y) < \hat{w}(\widehat{P}_i[a,u_i]) + d_H(a,y) = 2x_i +d_H(a,y) = d_H(u_i,y)$, a contradiction. 

Finally, we focus on $\widehat{P}_i[u_i,v_i]$, which has length $2L+4$. Again, the tool here is the anti-Monge inequality, and specifically, we will show that:
\begin{equation*}
    d_H(u_i,x) + d_H(v_i,y) > d_H(u_i,v_i) + d_H(x,y) 
\end{equation*}
Observe by \Cref{eq:ui-z-inJ}, we have:
\begin{equation*}
    \begin{split}
         d_H(u_i,x) + d_H(v_i,y)  &= 2x_i + d_J(a,x) + 2x_i + d_J(b,y)\\
         &= 4x_i + d_H(a,x) + d_H(b,y)\\
         &\geq 4x_i  + d_H(a,b) - d_H(x,y) \qquad \text{(triangle inequality)}\\
         &= 4x_i  +  2L+4  - d_H(x,y)\qquad \text{(by \Cref{eq:a-b-dist})}\\
         &> 4(B/4)+  2L+4  - d_H(x,y) \qquad \text{(by \Cref{eq:xi-B-L})}\\
         &\geq 2L+4   + d_H(x,y) \qquad \text{(since $(x,y)$ is an edge of weight $\leq B/2$)}\\
         &=  d_H(u_i,v_i) + d_H(x,y) 
    \end{split}
\end{equation*}
as claimed.

\paragraph{5) Item paths are simple.} We want to show that for every $i\in [3m]$, $\widehat{P}_i$ is a simple path. Since $\widehat{P}_i[a,v_i]$ and $\widehat{P}_i[b,u_i]$ are shortest paths between their endpoints, there is no self-intersection. The only possibility left is that  $\widehat{P}_i[a,u_i]$  and $\widehat{P}_i[b,v_i]$, each of length exactly $2x_i$, intersect. Then this case can also be eliminated by the anti-Monge inequality, as follows:
\begin{equation*}
\begin{split}
    d_H(a,v_i) + d_H(b,u_i) &=  (2L+4+2x_i) +    (2L+4+2x_i) \qquad \text{(by \Cref{eq:a-vi-b-u_i-dist})}\\
    &> 4x_i = d_H(a,u_i)  + d_H(b,v_i)
\end{split}
\end{equation*}
as desired.

\section{Exact algorithm for Metric Planarity Embedding}

 One can think of a metric $(X,d_X)$ as an (edge-weighted) complete graph $K_{|X|}$ with vertex set $X$. From this point of view, embedding a planar metric can be seen as a two-step process: (i) drawing a complete graph on the plane and (ii) assigning weights to the edges of the drawing.  Step (i) is the key difficulty; implementing step (ii) is rather standard, using linear programming, e.g.,~\cite{ChenTan2025,LeTiernoYang2026}.  Clearly, not all planar drawings of the complete graph $K_{|X|}$ can be assigned weights to realize the metric distances (otherwise, the problem is in P); such a drawing is called \EMPH{infeasible}. Here we seek to enumerate an (exponentially large) collection of all unique planar drawings such that at least one drawing in the collection is feasible. Fortunately, the pair of pioneering papers by Kynčl~\cite{kynvcl2009enumeration,kynvcl2013improved} has laid much groundwork for this step.  A central concept is a \EMPH{simple drawing} of a complete graph $K_n$.
 
\begin{definition}\label{def:feasible-drawing}
A \EMPH{simple drawing} $\Dr$ of the complete graph with vertex set $X$ is a drawing on the plane where: 
\begin{itemize}
     \item between any two pair of vertices $x,y\in X$, there exists a simple curve $\pi_\Dr(x,y)$ joining $x$ to $y$,
     \item any two curves intersect at most once: the intersection point is either a common endpoint or a proper crossing,
     \item no three curves cross at the same point. 
\end{itemize}
\end{definition}

Vertices in $X$ are called \EMPH{terminals} of the drawing $\Dr$.  The  \EMPH{graph of a simple drawing} $\Dr$ is a planar graph $G_{\Dr}$ that has a vertex corresponding to each terminal and each crossing point of any two curves. For each curve joining a pair of terminals, there is an edge in $G_{\Dr}$ between every two consecutive vertices on the curve. For any assignment of weights to the edges of $G_{\Dr}$, and pair of terminals $x,y\in X$, we shall let $w(\pi_{\Dr}(x,y))$ denote the sum of the weights of all the edges on the path $\pi_{\Dr}(x,y)$.

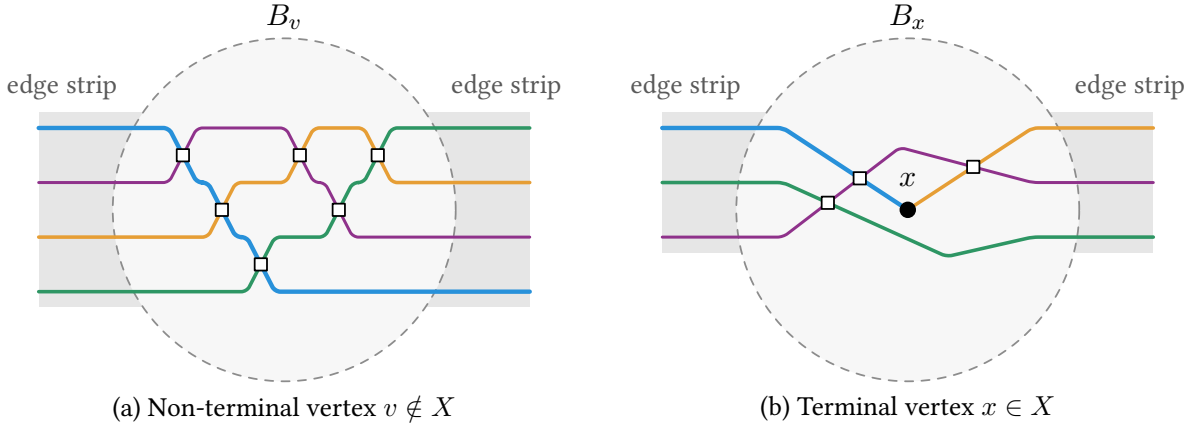
\begin{figure}[!htb]
  \centering
  \begingroup
  \definecolor{PMRblue}{RGB}{40,145,218}
  \definecolor{PMRpurple}{RGB}{143,50,139}
  \definecolor{PMRorange}{RGB}{230,158,53}
  \definecolor{PMRgreen}{RGB}{46,150,100}
  \resizebox{.96\linewidth}{!}{%
  \begin{tikzpicture}[
    x=1cm,y=1cm,font=\normalsize,
    line cap=round,line join=round,
    blue curve/.style={draw=PMRblue,line width=1.6pt,rounded corners=2pt},
    purple curve/.style={draw=PMRpurple,line width=1.1pt,rounded corners=2pt},
    orange curve/.style={draw=PMRorange,line width=1.3pt,rounded corners=2pt},
    green curve/.style={draw=PMRgreen,line width=1.3pt,rounded corners=2pt},
    proper crossing/.style={rectangle,draw=black,fill=white,
      line width=.65pt,minimum size=4.5pt,inner sep=0pt},
    terminal/.style={circle,fill=black,minimum size=6pt,inner sep=0pt},
    disk boundary/.style={draw=black!45,line width=.65pt,dashed},
    panel label/.style={font=\small,align=center}
  ]
    \begin{scope}
      \fill[black!10] (-3.15,-1.25) rectangle (-1.6,1.25);
      \fill[black!10] (1.6,-1.25) rectangle (3.15,1.25);
      \fill[black!3] (0,0) circle[radius=2.2];
      \draw[disk boundary] (0,0) circle[radius=2.2];
      \draw[blue curve] (-3.15,1.05)--(-1.48,1.05)--(-1.12,.35)
        --(-.98,.35)--(-.62,-.35)--(-.48,-.35)--(-.12,-1.05)
        --(3.15,-1.05);
      \draw[purple curve] (-3.15,.35)--(-1.48,.35)--(-1.12,1.05)
        --(.02,1.05)--(.38,.35)--(.52,.35)--(.88,-.35)--(3.15,-.35);
      \draw[orange curve] (-3.15,-.35)--(-.98,-.35)--(-.62,.35)
        --(.02,.35)--(.38,1.05)--(1.02,1.05)--(1.38,.35)--(3.15,.35);
      \draw[green curve] (-3.15,-1.05)--(-.48,-1.05)--(-.12,-.35)
        --(.52,-.35)--(.88,.35)--(1.02,.35)--(1.38,1.05)--(3.15,1.05);
      \foreach \cx/\cy in {-1.3/.7,-.8/0,-.3/-.7,.2/.7,.7/0,1.2/.7}{
        \node[proper crossing] at (\cx,\cy) {};
      }
      \node at (0,2.48) {$B_v$};
      \node[font=\small,black!65] at (-2.85,1.58) {edge strip};
      \node[font=\small,black!65] at (2.85,1.58) {edge strip};
      \node[panel label] at (0,-2.55) {(a) Non-terminal vertex $v\notin X$};
    \end{scope}

    \begin{scope}[xshift=8cm]
      \fill[black!10] (-3.15,-.55) rectangle (-1.6,1.25);
      \fill[black!10] (1.6,-.55) rectangle (3.15,1.25);
      \fill[black!3] (0,0) circle[radius=2.2];
      \draw[disk boundary] (0,0) circle[radius=2.2];
      \draw[blue curve] (-3.15,1.05)--(-1.6,1.05)--(0,0);
      \draw[orange curve] (3.15,1.05)--(1.6,1.05)--(0,0);
      \draw[purple curve] (-3.15,-.35)--(-1.6,-.35)
        --(-.1,.8)--(1.6,.35)--(3.15,.35);
      \draw[green curve] (-3.15,.35)--(-1.6,.35)
        --(.5,-.6)--(1.6,-.35)--(3.15,-.35);
      \foreach \cx/\cy in {-1.02578/.09024,-.61611/.40432,.83996/.55122}{
        \node[proper crossing] at (\cx,\cy) {};
      }
      \node[terminal] at (0,0) {};
      \node[above=5pt] at (0,0) {$x$};
      \node at (0,2.48) {$B_x$};
      \node[font=\small,black!65] at (-2.85,1.58) {edge strip};
      \node[font=\small,black!65] at (2.85,1.58) {edge strip};
      \node[panel label] at (0,-2.55) {(b) Terminal vertex $x\in X$};
    \end{scope}
  \end{tikzpicture}%
  }
  \endgroup
  \caption{Routing strands inside a vertex disk. White squares mark proper crossings: each pair crosses at most once
    inside the disk, and no three curves cross at a point. Strands in each edge strip are disjoint.}
  \label{fig:simple-drawing}
\end{figure}

The following lemma establishes the existence of a simple drawing whose graph preserves the metric distances. The basic idea is to duplicate edges shared by multiple shortest paths and make the paths cross inside small disks around the vertices. Similar routing ideas were also developed previously \cite{CEX14,pupyrev2016edge}, but in different contexts. 

\begin{lemma}\label{lm:simple-drawing}
If $(X,d_X)$ is planar, then there exists a simple
drawing $\Dr$ of $K_{|X|}$ and nonnegative edge weights on its graph
$G_{\Dr}$ such that $d_{G_{\Dr}}(x,y)
    = w(\pi_{\Dr}(x,y))
    = d_X(x,y)$, for every $x,y\in X$.
\end{lemma}
\begin{proof}
Let $H$ be a planar embedding of $(X,d_X)$. We construct a simple drawing $\Dr$ of $K_{|X|}$ as follows.  First, we replace each vertex $v\in V(H)$
by a small disk $B_v$, and each edge $e\in E(H)$ by a narrow strip between its
endpoint disks, such that distinct edges have disjoint strips. For every pair of terminals $x,y\in X$, we choose a shortest $x$-$y$ path $\pi_H(x,y)$, and draw a curve $\pi_{\Dr}(x,y)$ alongside the path $\pi_H(x,y)$ such that inside each edge strip, the \EMPH{strands} of distinct curves are disjoint. Thus, if an edge $e$ belongs to $k$ different shortest paths, its strip has $k$ strands. See \Cref{fig:simple-drawing}. We now describe how to construct the crossings inside the vertex disks.

Consider a non-terminal vertex $v\not\in X$. We connect each incoming strand at $v$ to its outgoing strand in such a way that two curves cross at most once, and no three curves cross at the same point; see \Cref{fig:simple-drawing}(a). For a terminal vertex $x\in X$, we create a node of $\Dr$ at the center of the disk $B_x$ and connect all curves that have $x$ as an endpoint to this node. For all other curves (that pass through $x$), by the definition of a simple drawing, we have to reroute them to avoid $x$. We do so by connecting their incoming strand to their outgoing strand while avoiding the center of $B_x$, and ensuring that two curves cross at most once, and no three curves cross at the same point.  See \Cref{fig:simple-drawing}(b). These curves give a simple drawing $\Dr$ of $K_{|X|}$.

Now consider an edge $e$ of $G_{\Dr}$. If the edge $e$ completely lies inside some vertex disk $B_v$, we give it a weight of $0$. Else, the edge $e$ corresponds to a sequence of contiguous strands of some curve $\pi_{\Dr}(x,y)$ between two vertices of $\pi_H(x,y)$. We define the weight of $e$ to be the length of the subpath of $\pi_H(x,y)$ between these vertices. Hence we have that $w(\pi_{\Dr}(x,y))=d_H(\pi_H(x,y))=d_H(x,y)$.
Moreover, any walk from $x$ to $y$ in $G_{\Dr}$ can be traced in $H$ by
concatenating the corresponding subpaths, without changing its length.
Therefore, every $x$-$y$ walk in $G_{\Dr}$ has length at least
$d_H(x,y)=d_X(x,y)$. Therefore, it holds that $d_{G_{\Dr}}(x,y)
    = w(\pi_{\Dr}(x,y))
    = d_X(x,y)$, for every $x,y\in X$, and we are done.
\end{proof}

Due to the above lemma, it suffices to look at simple drawings of the complete graph on $|X|$ vertices. Our strategy would be to enumerate all simple drawings of $K_{|X|}$, and for each such drawing, decide whether a feasible edge-weight assignment exists that realizes all the pairwise metric distances. To this end, we define the notion of the \EMPH{rotation} at a terminal, and the \EMPH{rotation system} of a simple drawing.

\begin{definition}[Rotation and Rotation system]
     Given a simple drawing $\Dr$ of $K_{|X|}$, the \EMPH{rotation} of a terminal $x\in X$ is the clockwise cyclic order of the curves incident to $x$ in the drawing $\Dr$. The \EMPH{rotation system} of $\Dr$ is the set of rotations of all the terminals.
\end{definition}

It is well-known that the rotation system uniquely determines which pairs of curves in a simple drawing cross \cite{pach2006many,kynvcl2009enumeration}. However, it is a priori not clear whether all crossing orders would yield graphs that can realize the metric distances. We shall show in \Cref{lm:triangle-flip} that this is indeed the case, and hence it would suffice to consider any arbitrary simple drawing realizing a particular rotation system. For this, we first define the notion of \EMPH{triangle flips} in a simple drawing. Let $\Dr$ be a simple drawing and consider three curves that pairwise cross each other in $\Dr$. A triangle flip is a local modification in which one of three curves is moved across the crossing of the other two, while keeping the drawing unchanged outside the triangular region bounded by the three crossings (see \Cref{fig:triangle-flip-unweighted}). We shall use the following result of Gioan \cite{Gioan2022}.

\begin{lemma}[Gioan\cite{Gioan2022}]
\label{lem:gioan}
  Any two simple drawings with the same rotation system can be transformed into each other by
a sequence of triangle flips.
\end{lemma}

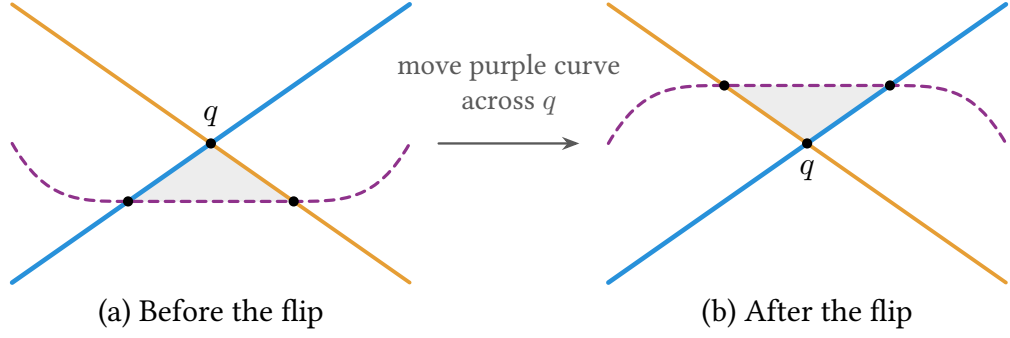
\begin{figure}[!htb]
  \centering
  \begingroup
  \definecolor{flipblue}{RGB}{40,145,218}
  \definecolor{flippurple}{RGB}{143,50,139}
  \definecolor{fliporange}{RGB}{230,158,53}
  \resizebox{.8\linewidth}{!}{%
  \begin{tikzpicture}[
    x=1cm,y=1cm,font=\normalsize,
    line cap=round,line join=round,
    fixed blue/.style={draw=flipblue,line width=1.6pt},
    fixed orange/.style={draw=fliporange,line width=1.3pt},
    moving curve/.style={draw=flippurple,line width=1.1pt,
      dash pattern=on 3pt off 2.5pt},
    crossing/.style={circle,fill=black,inner sep=0pt,minimum size=3.5pt}
  ]
    \begin{scope}
      \fill[black!7] (1.4,-.7)--(2.4,0)--(3.4,-.7)--cycle;
      \draw[fixed blue] (0,-1.68)--(4.8,1.68);
      \draw[fixed orange] (0,1.68)--(4.8,-1.68);
      \draw[moving curve] (0,0)
        .. controls (.4,-.7) and (.8,-.7) .. (1.4,-.7)
        --(3.4,-.7)
        .. controls (4,-.7) and (4.4,-.7) .. (4.8,0);
      \node[crossing] at (1.4,-.7) {};
      \node[crossing] at (3.4,-.7) {};
      \node[crossing,label=above:$q$] at (2.4,0) {};
      \node at (2.4,-2.05) {(a) Before the flip};
    \end{scope}

    \draw[-{Stealth[length=6pt]},black!65,line width=.8pt]
      (5.15,0)--(6.85,0)
      node[midway,above=6pt,font=\small,align=center]
        {move purple curve\\across $q$};

    \begin{scope}[xshift=7.2cm]
      \fill[black!7] (1.4,.7)--(2.4,0)--(3.4,.7)--cycle;
      \draw[fixed blue] (0,-1.68)--(4.8,1.68);
      \draw[fixed orange] (0,1.68)--(4.8,-1.68);
      \draw[moving curve] (0,0)
        .. controls (.4,.7) and (.8,.7) .. (1.4,.7)
        --(3.4,.7)
        .. controls (4,.7) and (4.4,.7) .. (4.8,0);
      \node[crossing] at (1.4,.7) {};
      \node[crossing] at (3.4,.7) {};
      \node[crossing,label=below:$q$] at (2.4,0) {};
      \node at (2.4,-2.05) {(b) After the flip};
    \end{scope}
  \end{tikzpicture}%
  }
  \endgroup
  \caption{A triangle flip. The purple curve moves across
    the crossing $q$ of the two fixed curves. Each pair of curves still
    crosses exactly once, but the order of the crossings along each curve
    reverses.}
  \label{fig:triangle-flip-unweighted}
\end{figure}

The above result concerns transformations between simple drawings with the same rotation system, without imposing metric constraints. The next lemma shows that feasibility of realizing $(X,d_X)$ by non-negative edge weights is also preserved under these transformations.

\begin{lemma}\label{lm:triangle-flip}
Let $\Dr$ and $\Dr'$ be simple drawings of $K_{|X|}$ with the same rotation
system, and $G_{\Dr}$ and $G_{\Dr'}$ their graphs, respectively. 
Then $G_{\Dr}$ admits nonnegative edge lengths such that $d_{G_{\Dr}}(x,y)
    = w(\pi_{\Dr}(x,y))
    = d_X(x,y)$, for all $x,y\in X$,
if and only if $G_{\Dr'}$ admits such edge lengths.
\end{lemma}

\begin{proof}
Due to \Cref{lem:gioan}, it suffices to consider the case when $\Dr'$ and $\Dr$ are related by a single triangle flip.

Let $p,q,r$ be the three crossings bounding a triangle in the drawing $\Dr$ before the
flip, and let $a,b,c$ denote the weights of the edges joining $pq$, $pr$, $qr$ in the graph $G_{\Dr}$, respectively. Since each curve in $\Dr$ corresponds to a shortest path between two terminals in $G_{\Dr}$, it follows that
\[
    a\leq b+c,\qquad b\leq a+c,\qquad c\leq a+b.
\]
Define the non-negative quantities
\[
    \lambda_p=\frac{a+b-c}{2},\qquad
    \lambda_q=\frac{a+c-b}{2},\qquad
    \lambda_r=\frac{b+c-a}{2}.
\]
and observe that
\begin{equation}
\label{eqn:abc}
    a=\lambda_p+\lambda_q,\qquad
    b=\lambda_p+\lambda_r,\qquad
    c=\lambda_q+\lambda_r.
\end{equation}

In the new drawing $\Dr'$ obtained after the flip, we give the three new triangle sides length
zero. We associate each of the six exterior edges with the crossing \(s\in\{p,q,r\}\) it was incident on before the flip, and increase its length by
\(\lambda_s\), leaving all other edge lengths unchanged; see \Cref{fig:triangle-flip}.

\begin{figure}[!htb]
  \centering
  \begingroup
  \definecolor{PMRblue}{RGB}{40,145,218}
  \definecolor{PMRpurple}{RGB}{143,50,139}
  \definecolor{PMRorange}{RGB}{230,158,53}
  \definecolor{PMRgray}{RGB}{235,235,235}

  \resizebox{.92\linewidth}{!}{%
  \begin{tikzpicture}[
    x=1cm,y=1cm,font=\normalsize,
    line cap=round,line join=round,
    blue path/.style={draw=PMRblue,line width=1.6pt},
    purple path/.style={draw=PMRpurple,line width=1.1pt,
      dash pattern=on 3pt off 2.5pt},
    orange path/.style={draw=PMRorange,line width=1.3pt},
    crossing vertex/.style={rectangle,draw=black,fill=white,
      line width=.8pt,minimum size=5.5pt,inner sep=0pt},
    length label/.style={fill=white,inner sep=1.5pt},
    panel label/.style={font=\small,align=center}
  ]

    \begin{scope}
      \coordinate (p) at (1.4,-.6);
      \coordinate (q) at (2.6,.6);
      \coordinate (r) at (3.8,-.6);
      \fill[PMRgray] (p)--(q)--(r)--cycle;
      \draw[blue path] (0,-1.4)
        .. controls (.65,-1.4) and (1.0,-1.0) .. (p)
        --(q) .. controls (3.4,1.4) and (4.4,1.4) .. (5.2,1.4);
      \draw[purple path] (0,0)
        .. controls (.65,0) and (.9,-.6) .. (p)
        --(r) .. controls (4.3,-.6) and (4.55,0) .. (5.2,0);
      \draw[orange path] (0,1.4)
        .. controls (.8,1.4) and (1.8,1.4) .. (q)
        --(r) .. controls (4.2,-1.0) and (4.55,-1.4) .. (5.2,-1.4);
      \foreach \s in {p,q,r}{\node[crossing vertex] at (\s) {};}
      \node[below=5pt] at (p) {$p$};
      \node[above=5pt] at (q) {$q$};
      \node[below=5pt] at (r) {$r$};
      \node[PMRblue,length label] at (1.83,.19) {$a$};
      \node[PMRpurple,length label] at (2.6,-.82) {$b$};
      \node[PMRorange,length label] at (3.37,.19) {$c$};
      \node[panel label] at (2.6,-2.0) {(a) Before the flip: $G_{\Dr}$};
    \end{scope}

    \draw[-{Stealth[length=6pt]},black!65,line width=.8pt]
      (5.65,0)--(6.85,0)
      node[midway,above=5pt,font=\small,align=center] {triangle\\flip};

    \begin{scope}[xshift=7.5cm]
      \coordinate (rp) at (1.4,.6);
      \coordinate (qp) at (2.6,-.6);
      \coordinate (pp) at (3.8,.6);
      \fill[PMRgray] (rp)--(qp)--(pp)--cycle;
      \draw[blue path] (0,-1.4)
        .. controls (.8,-1.4) and (1.8,-1.4) .. (qp)
        --(pp) .. controls (4.2,1.0) and (4.55,1.4) .. (5.2,1.4);
      \draw[purple path] (0,0)
        .. controls (.65,0) and (.9,.6) .. (rp)
        --(pp) .. controls (4.3,.6) and (4.55,0) .. (5.2,0);
      \draw[orange path] (0,1.4)
        .. controls (.65,1.4) and (1.0,1.0) .. (rp)
        --(qp) .. controls (3.4,-1.4) and (4.4,-1.4) .. (5.2,-1.4);
      \foreach \s in {pp,qp,rp}{\node[crossing vertex] at (\s) {};}
      \node[above=5pt] at (rp) {$r'$};
      \node[below=5pt] at (qp) {$q'$};
      \node[above=5pt] at (pp) {$p'$};
      \node[PMRorange,length label] at (1.83,-.17) {$0$};
      \node[PMRpurple,length label] at (2.6,.84) {$0$};
      \node[PMRblue,length label] at (3.37,-.17) {$0$};

      \node[PMRorange,length label] at (.45,1.64) {$+\lambda_q$};
      \node[PMRpurple,length label] at (.25,.58) {$+\lambda_p$};
      \node[PMRblue,length label] at (.45,-1.72) {$+\lambda_p$};
      \node[PMRblue,length label] at (4.75,1.64) {$+\lambda_q$};
      \node[PMRpurple,length label] at (4.95,.58) {$+\lambda_r$};
      \node[PMRorange,length label] at (4.75,-1.72) {$+\lambda_r$};
      \node[panel label] at (2.6,-2.0) {(b) After the flip: $G_{\Dr'}$};
    \end{scope}
  \end{tikzpicture}%
  }
  \endgroup
  \caption{A triangle flip and the length reassignment in
    \Cref{lm:triangle-flip}. The old triangle sides have lengths
    $a=\lambda_p+\lambda_q$, $b=\lambda_p+\lambda_r$, and
    $c=\lambda_q+\lambda_r$; the three new sides have length zero.
    Each exterior edge gains $\lambda_s$ according to its incidence
    at $s$ before the flip.
    Thus each colored path retains its total length.}
  \label{fig:triangle-flip}
\end{figure}
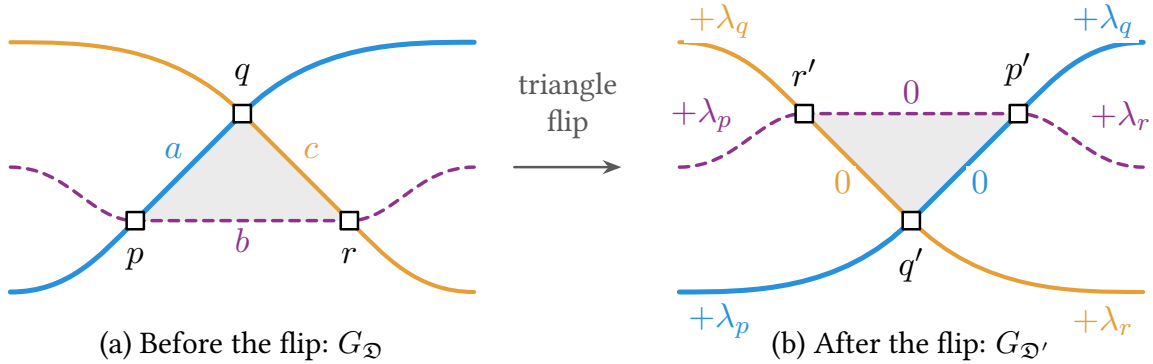

We show that the distances between any pair of terminals do not decrease in $G_{\Dr'}$.
Consider a walk in $G_{\Dr'}$ through the flipped region that enters through an
exterior edge associated with \(p\). If it leaves through an exterior edge
associated with \(q\), its added length is \(\lambda_p+\lambda_q=a\),
exactly the length of the original triangle side \(pq\). If it instead leaves
through an exterior edge associated with \(r\), its added length is
\(\lambda_p+\lambda_r=b\), exactly the length of \(pr\), by \eqref{eqn:abc}.
If it leaves through an exterior edge associated with \(p\), its added length
is \(2\lambda_p\geq 0\), so this walk is no shorter than staying at \(p\).
The same argument applies when the walk enters through an exterior edge
associated with \(q\) or \(r\). Thus every walk after the flip can be converted
into a walk before the flip of no greater length.

On the other hand, every shortest path affected by the flip
previously used one triangle side, say \(pq\), of length
\(a\). After the flip it uses exterior edges associated with \(p\) and \(q\),
and therefore gains a length of \(
  \lambda_p+\lambda_q=a,
\)
which is exactly the length it lost. Hence every shortest path in $G_{\Dr'}$ still has
length \(d_X(x,y)\). Therefore we obtain that \(
  d_{G_{\Dr'}}(x,y)
  =w(\pi_{\Dr'}(x,y))
  =d_X(x,y),
\)
for every $x,y\in X$, and we are done.
\end{proof}

Let $n=|X|$. As discussed earlier, we shall enumerate all feasible rotation systems on $n$ terminals, and check whether there exists a feasible assignment of weights to the edges of the corresponding simple drawings that realizes the pairwise distances of the metric $(X,d_X)$. Since there are $(n-2)!$ possible rotations at each terminal, a naive enumeration bounds the number of rotation systems as $((n-2)!)^n=2^{\Theta(n^2\log n)}$. A recent work of Kyn\v{c}l \cite{kynvcl2013improved} improves this bound to $2^{n^2\alpha(n)^{O(1)}}$, where $\alpha(n)$ denotes the inverse Ackermann function. Although the proof in \cite{kynvcl2013improved} only gives an existential bound, we show that it is also possible to enumerate the feasible rotation systems in asymptotically the same running time. For this, we shall need the following result of an earlier work of Kyn\v{c}l \cite{Kyncl2011}.

\begin{lemma}[Kyn\v{c}l\cite{Kyncl2011}]
\label{lem:decide-poly-time}
  Given a rotation system, it is possible to decide in polynomial time whether there exists a simple drawing realizing it, and if yes, output such a drawing. 
\end{lemma}

We now show the following lemma.

\begin{lemma}\label{lm:enumerate-rotations}
     All the feasible rotation systems of the complete graph on $n$ terminals can be enumerated in time $2^{n^2\alpha(n)^{O(1)}}$.
\end{lemma}
\begin{proof}
For $k\in [n]$, let $\mathcal{R}_k$ denote the set of rotation systems on the labeled
vertex set $[k]$ that are realizable by simple drawings of the complete graph on $k$ vertices. From \cite{kynvcl2013improved}, we have $|\mathcal{R}_k|=2^{k^2\alpha(k)^{O(1)}}$.

We enumerate the sets $\mathcal{R}_3,\ldots,\mathcal{R}_n$
incrementally, starting with the unique rotation system on three
vertices. Given $R\in\mathcal{R}_k$, generate every rotation system
$R'$ on $[k+1]$ whose restriction to $[k]$ is $R$, as follows:
\begin{itemize}
    \item At each vertex $v\in[k]$, insert the label $k+1$ into
    one of the $k-1$ gaps of the rotation at $v$.
    \item Choose a rotation of $[k]$ at vertex $k+1$.
\end{itemize}
There are exactly $(k-1)^k(k-1)!
        =2^{O(k\log k)}$
such extensions for each $R\in \mathcal{R}_k$. For each extension, we test whether there exists a simple drawing realizing it, and retain only the feasible ones. This test can be done in $k^{O(1)}$ time by \Cref{lem:decide-poly-time}.

Every retained system is feasible by construction. Conversely,
if $R'\in\mathcal{R}_{k+1}$, deleting vertex $k+1$ from a realizing
drawing shows that its restriction to $[k]$ gives a feasible rotation system on $k$ vertices. Hence the procedure
generates $R'$ from its restriction. Moreover, this restriction is
unique, and the insertion positions and cyclic order at the new
vertex are uniquely determined by $R'$. Thus every feasible labeled
rotation system on $k+1$ vertices is generated exactly once.

It follows that the total time required to enumerate all the feasible rotation systems on $n$ terminals is at most 
\[ \sum_{k=3}^{n-1} \left(|\mathcal{R}_k|\cdot 2^{O(k\log k)}\cdot k^{O(1)}\right) \le \sum_{k=3}^{n-1} \left(2^{k^2\alpha(k)^{O(1)}}\cdot 2^{O(k\log k)}\cdot k^{O(1)}\right) = 2^{n^2\alpha(n)^{O(1)}},\]
and we are done.
\end{proof}

Finally, it remains to decide whether a given simple drawing can realize the metric $(X,d_X)$ using linear programming. We note that this step is standard, see, e.g.,~\cite{ChenTan2025,LeTiernoYang2026}, but we include its (short) detail for completeness.

\begin{lemma}\label{lm:drawing-feasibility}
     Given a simple drawing $\Dr$ on $n$ terminals, it is possible to decide in time $(n\log \Delta)^{O(1)}$ whether there exists an assignment of non-negative weights to the edges of the graph $G_{\Dr}$ such that for any $x,y\in X$, it holds that $d_{G_{\Dr}}(x,y)
    = w(\pi_{\Dr}(x,y))
    = d_X(x,y)$.
\end{lemma}
\begin{proof}
     We introduce non-negative variables $\{\ell_e\}_{e\in E(G_{\Dr})}$, which shall stand for the weights of the edges of $G_{\Dr}$, and a non-negative variable $\phi(u,v)$ for every pair of vertices $u,v \in V(G_{\Dr})$, representing the length of the shortest path between $u$ and $v$. We formulate the following linear program. The most interesting inequality is the third, which is essentially the triangle inequality.      
\begin{align*}
  \sum_{e\in \pi_{\Dr}(x,y)}\ell_e&=d_X(x,y)
    &&\text{for each pair $x,y\in X$},\\
  \phi(x,y)&=d_X(x,y)
    &&\text{for each pair $x,y\in X$},\\
  -\ell_{uv}\leq\phi(x,v)-\phi(x,u)&\leq\ell_{uv}
    &&\text{for each $x\in X$ and $(u,v)\in E(G_{\Dr})$},\\
    \phi(v,v) &= 0 &&\text{for each $v\in V(G_{\Dr})$},\\
    \ell_e, \phi(u,v) &\geq 0 &&\text{for each $e\in E(G_{\Dr})$, and each pair $u,v \in V(G_{\Dr})$}.
\end{align*}

Given a feasible solution to the above linear program, consider any pair of terminals $x,y\in X$ and let $x=v_0,v_1,\ldots, v_k=y$ be any walk from $x$ to $y$ in the graph $G_{\Dr}$. Then from the LP constraints, we have that
\[ \sum_{i=1}^{k} \ell_{v_{i-1}v_{i}} \ge \sum_{i=1}^{k} (\phi(x,v_{i})-\phi(x,v_{i-1}))=\phi(x,y)-\phi(x,x)=\phi(x,y) = d_X(x,y) \]
and thus any walk has length at least $d_X(x,y)$. Also the path $\pi_{\Dr}(x,y)$ achieves the desired length $d_X(x,y)$. Therefore $\{\ell_e\}_{e\in E(G_{\Dr})}$ forms a feasible assignment of weights to the edges such that $d_{G_{\Dr}}(x,y)
    = w(\pi_{\Dr}(x,y))
    = d_X(x,y)$. 

Conversely, suppose there exists a feasible assignment of non-negative weights $\{w_e\}_{e\in E(G_{\Dr})}$ to the edges of $G_{\Dr}$ that realizes the metric distances in $X$. Then it is easy to see that setting $\ell_e := w_e$ for each $e\in E(G_{\Dr})$, and $\phi(u,v)$ to be equal to the shortest path distance between the vertices $u$ and $v$ in the weighted graph $G_{\Dr}$, gives a feasible solution to the linear program.

Since $|V(G_{\Dr})|=O(n^4)$, the linear program has only $O(n^8)$ variables and constraints, and therefore its feasibility can be determined in $(n\log \Delta)^{O(1)}$ time.
\end{proof}

We are now ready to prove \Cref{thm:algorithm}.
\begin{proof}[Proof of \Cref{thm:algorithm}]
We enumerate all feasible rotation systems on $X$ using
\Cref{lm:enumerate-rotations}. For each system, we construct a corresponding
simple drawing $\Dr$ in $n^{O(1)}$ time using the algorithm
of~\cite{Kyncl2011}, and apply \Cref{lm:drawing-feasibility} to its graph $G_{\Dr}$. If any of these tests succeed, we return the corresponding weighted graph.

Every successful test yields a weighted planar graph realizing $(X,d_X)$.
Conversely, if $(X,d_X)$ has a planar embedding, then
\Cref{lm:simple-drawing} gives a simple drawing whose graph admits
the required edge weights. Its rotation system is enumerated, and
\Cref{lm:triangle-flip} guarantees that the simple drawing constructed for this
system also admits such weights. 

By \Cref{lm:enumerate-rotations}, the enumeration takes
$2^{n^2\alpha(n)^{O(1)}}$ time.
For each rotation system, constructing the simple drawing and applying
\Cref{lm:drawing-feasibility} takes $(n\log\Delta)^{O(1)}$ time.
The total running time is therefore
\[
  2^{n^2\alpha(n)^{O(1)}}(n\log\Delta)^{O(1)}
  =2^{n^2\alpha(n)^{O(1)}}(\log\Delta)^{O(1)},
\]
as claimed.
\end{proof}

\section{Faster Approximate Metric Planarity Testing}

In this section, we prove \Cref{thm:approximate-algo}. For this, we use the following result of Chang, Krauthgamer and Tan \cite{chang2022almost} on near-linear sized emulators of planar graphs.

\begin{theorem}[Chang, Krauthgamer and Tan\cite{chang2022almost}]\label{thm:planar-emulator}
  For every planar graph $G$ with $n$ distinguished terminal vertices $T$, and $\epsilon \in (0,1)$, there exists another planar graph $G'$ with $T\subseteq V(G')$ of size $|V(G')|=n(\log n/\epsilon)^{O(1)}$, satisfying
  \[ \forall x,y\in T, \qquad d_G(x,y)\le d_{G'}(x,y)\le (1+\epsilon)d_G(x,y) \]
\end{theorem}

Our algorithm uses only the existence of these small emulators and proceeds by enumerating all such small emulators. To make this enumeration efficient, we also need to restrict the possible edge weights to a small range. This is possible since we can tolerate an error in the distortion, and so we simply round up the distances when they are too small.

\begin{claim}\label{lm:rounding}
  If a metric $(X,d_X)$ has a $(1+\epsilon)$-approximate planar embedding $G$ with $q$ vertices, then it has a $(1+2\epsilon)$-approximate planar embedding  
  $G'$ on the same vertex set where all edge weights belong to a set of size $O(n^2q/\epsilon)$.
\end{claim}
\begin{proof}
Let $G$ be a $(1+\epsilon)$-approximate planar embedding of $(X,d_X)$, and for each pair of distinct terminals $x,y\in X$, we choose a shortest path $P_{xy}$ in $G$, breaking ties arbitrarily.
We delete all edges of $G$ that do not belong to these
paths. For each remaining edge $e$, let 
\[ \beta_e := \min\{d_X(x,y) \mid x,y\in X \text{ and } e\in P_{xy}\}, \]
i.e., $\beta_e$ is the smallest value of $d_X(x,y)$ across all pairs of terminals $x,y\in X$ for which the chosen path $P_{xy}$ contains the edge $e$. Clearly, the weight of the edge $e$ must be at most $(1+\epsilon)\beta_e$. We round this weight up to the nearest non-negative integer
multiple of $\epsilon \beta_e/q$, obtaining the planar graph $G'$ whose weights
belong to the set
\[
  W=\left\{\frac{j\epsilon d_X(x,y)}{q}:
    x,y\in X,\ j\in\mathbb{N},\
    j\le \left\lceil \frac{(1+\epsilon)q}{\epsilon}\right\rceil\right\}.
\]

Clearly $|W|=O(n^2q/\epsilon)$ and can be enumerated from the input
in $(nq\log \Delta/\epsilon)^{O(1)}$ time.

Rounding only increases edge weights, so $d_{G'}(x,y)\ge d_G(x,y)\ge d_X(x,y)$.
On the other hand, every edge $e\in P_{xy}$ has $\beta_e\le d_X(x,y)$,
so its weight increases by at most $\epsilon d_X(x,y)/q$.
Since $P_{xy}$ has at most $q$ edges, its total length increases by at
most $\epsilon d_X(x,y)$. Hence
\[
  d_{G'}(x,y)\le d_G(x,y)+\epsilon d_X(x,y)
  \le (1+2\epsilon)d_X(x,y),
\]
and we are done.
\end{proof}

We are now ready to prove \Cref{thm:approximate-algo}.
\begin{proof}[Proof of \Cref{thm:approximate-algo}]
Let $(X,d_X)$ be a metric that has a $(1+\epsilon/2)$-approximate planar embedding $G$, i.e.,
\[ d_X(x,y) \le d_G(x,y)\le \left(1+\frac{\epsilon}{2}\right)d_X(x,y) \qquad \text{for all } x,y\in X.\]

Combining \Cref{thm:planar-emulator} and \Cref{lm:rounding}, there exists a $(1+\epsilon)$-approximate planar embedding $G'$ with at most $q=n(\log n/\epsilon)^{O(1)}$ vertices and whose edge weights belong to a polynomial-time computable set $W$ of size $O(n^2q/\epsilon)$.

We label the terminals $1,2,\ldots,n$ and enumerate all planar graphs on at most $q$ vertices, and try every possible assignment of weights from $W$ to each enumerated graph $G$. We compute the lengths of the shortest paths between every pair of terminals and test whether 
\[
  d_X(x,y)\le d_G(x,y)\le (1+\epsilon)d_X(x,y)
  \qquad\text{for all }x,y\in X,
\]
and return any graph passing this test. If no graph passes the test, we return NO.

Since any planar graph has $O(q)$ edges, there are $q^{O(q)}$ candidate graphs and at most $|W|^{O(q)}$
weight assignments per graph. Computing all-pairs shortest paths in each graph takes $(q\log \Delta)^{O(1)}$ time. The total running time is therefore
\[
  q^{O(q)}\cdot |W|^{O(q)} \cdot (q\log \Delta)^{O(1)} 
  =2^{n(\log n/\epsilon)^{O(1)}}\log\Delta,
\]
as claimed.
\end{proof}

\bibliographystyle{plain}
\bibliography{references}
\appendix

\section{ETH Hardness of \ThreePartition}\label{app:three-partition}

In this section, we provide a self-contained proof of \cref{thm:ETH-3Partition}. 
Garey and Johnson's book~\cite[Section 4.2, Theorems 4.3 and 4.4, pp.~97--100]{GareyJohnson1979} 
proves NP-completeness of \ThreePartition{} using the chain of reductions
\[
\textsc{3-Dimensional Matching}\ \longrightarrow\ \textsc{4-Partition}\ \longrightarrow\ \ThreePartition.
\]
Although the first leg of this chain of reductions ensures that the input remains linear, the second
leg has a quadratic blow-up in the input. Hence this chain of reduction is not directly applicable
to get the desired ETH-bound. We instead show a reduction directly from \textsc{3-Dimensional Matching}. 

In the reduction for $\textsc{4-Partition}$ by Garey and Johnson~\cite{GareyJohnson1979}, the basic idea is to encode each (3D) matching edge $(i,j,k)$ in the \textsc{3-Dimensional Matching} instance as a quadruple $(x_i, y_j,z_k,t_{ijk})$ in the resulting  $4$-\textsc{Partition} instance, where the last number $t_{ijk}$ in the quadruple makes sure the sum is $B$. For   $3$-\textsc{Partition}, the reduction here must break up $x_i, y_j,z_k$ into \EMPH{two different triples} (the last number in each triple is used to guarantee the sum is $B$). To ``synchronize'' the two triples corresponding to a single 3D matching edge, the reduction must carefully add connector gadgets, which are the key difference from the reduction by Garey and Johnson~\cite{GareyJohnson1979}.

In \textsc{3-Dimensional Matching}, 
the input consists of
\begin{itemize}
  \item $U=\{u_1,\ldots,u_q\}$, $V=\{v_1,\ldots,v_q\}$, 
        and $W=\{w_1,\ldots ,w_q\}$, three disjoint sets of $q$ \EMPH{nodes} each,
  \item and a set of \EMPH{edges}, $E \subseteq U\times V\times W$, containing $m$ triples of nodes.
\end{itemize}
A \EMPH{perfect matching} is a subset $E'$ of edges $E$ such that 
every node in $U \cup V \cup W$ appears exactly once in $E'$.
The problem of \textsc{3-Dimensional Matching} asks to decide if there
exists a perfect matching for the input instance. 

Note that~\cite[Theorem~3.1]{Jansen2016} rules out
an algorithm for \textsc{3-Dimensional Matching} with running time 
$2^{o(m)}|I|^{\Oh(1)}$ under ETH, where $m=|E|$ and $|I|$ is the binary input length.
We may assume $m \ge 1$, $q \ge 2$ and that no node is isolated, since
all other instances are trivial to decide.

\Needspace{9\baselineskip}
\begin{lemma}\label{lem:matching-partition}
An instance of \textsc{3-Dimensional Matching} with $m$ edges
and no isolated nodes can be transformed in polynomial
time into an equivalent instance of \ThreePartition{} with $n=9m$
items and $B\leq n^{\Oh(1)}$.
\end{lemma}
\begin{proof}
  For a node $a\in U\cup V\cup W$, let $d(a)$ denote its degree,
  i.e.~the number of edges containing $a$.

  \subparagraph*{Construction of the reduced \ThreePartition{} instance.} 
  The reduction below essentially creates 9-dimensional vectors with 
  components lying in $\{-q,\ldots,q\}$, and represents these nine
  components together as a single integer using a large enough base $M$.
  
  We set $M=10(q+1)$, and we represent the nine components as coefficients
  of $M^0$ to $M^8$. The coefficient of $M^8$ is set to $1$ for all numbers 
  in order to keep them in the range $\left(B/4,B/2\right)$. The
  coefficients of $M^7$, $M^6$, $M^5$, and $M^4$ are used to restrict the 
  structure of the solution, in terms of which elements can appear in the
  same triple. The coefficients of $M^0$,
  $M^1$, and $M^2$ carry information about the exact nodes in $U$, $V$ and 
  $W$ the numbers represent. 

  For each node $a \in U \cup V \cup W$, we will create $d(a)$ many numbers in the \ThreePartition{}
  instance. The coefficient of $M^3$ distinguishes one number (which 
  we refer to as the \EMPH{flagged} number), of each node
  from its other $d(a)-1$ numbers. This flagged number has
  coefficient $1$ for a node in $U$ or $V$, and coefficient $-2$
  for a node in $W$; every other \EMPH{unflagged} number is identical
  to each other and has coefficient $0$.
  The flagged numbers will identify the edges of a perfect matching.
  Some of the other constructed numbers will also have nonzero 
  coefficients of $M^3$ so that the corresponding triples sum to the target $B$.
  
  Formally, the \ThreePartition{} instance is constructed as follows:
  \begin{itemize}

    \item (\EMPH{The $\mathtt{U}$-elements}) For every $u_i \in U$, we create $d(u_i)$ 
    many integers $\alpha_{i,1},\ldots,\alpha_{i,d(u_i)}$ defined as: 
    \[\alpha_{i,1}=M^8+M^4+M^3+i, \quad \text{and} \quad
      \alpha_{i,\ell}=M^8+M^4+i\quad(2\leq \ell\leq d(u_i)).\]
    We call these the \EMPH{$\mathtt{U}$-elements}. We refer to the specific collection 
    of $\alpha_{i,1}$ for all $i \in [q]$ as \EMPH{flagged $\mathtt{U}$-elements}
    and denote this collection as a multiset $A_U'$; we refer to the collection 
    of $\alpha_{i,\ell}$ with $2\leq \ell\leq d(u_i)$ for all $i \in [q]$ as 
    \EMPH{unflagged $\mathtt{U}$-elements} and denote this collection as a multiset $A_U$.

    \item (\EMPH{The $\mathtt{V}$-elements}) For every $v_j \in V$, we create $d(v_j)$ 
    many integers $\beta_{j,1},\ldots,\beta_{j,d(v_j)}$ defined as: 
    \[\beta_{j,1}=M^8+M^5+M^3+jM, \quad \text{and} \quad
      \beta_{j,\ell}=M^8+M^5+jM\quad(2\leq \ell\leq d(v_j)).\]
    We call these the \EMPH{$\mathtt{V}$-elements}. We refer to the specific collection 
    of $\beta_{j,1}$ for all $j \in [q]$ as \EMPH{flagged $\mathtt{V}$-elements}
    and denote this collection as a multiset $A_V'$; we refer to the collection 
    of $\beta_{j,\ell}$ with $2\leq \ell\leq d(v_j)$ for all $j \in [q]$ as 
    \EMPH{unflagged $\mathtt{V}$-elements} and denote this collection as a multiset $A_V$.

    \item (\EMPH{The $\mathtt{W}$-elements}) For every $w_k \in W$, we create $d(w_k)$ 
    many integers $\gamma_{k,1},\ldots,\gamma_{k,d(w_k)}$ defined as: 
    \[\gamma_{k,1}=M^8+M^6-2M^3+kM^2, \quad \text{and} \quad
      \gamma_{k,\ell}=M^8+M^6+kM^2\quad(2\leq \ell\leq d(w_k)).\]
    We call these the \EMPH{$\mathtt{W}$-elements}. We refer to the specific collection 
    of $\gamma_{k,1}$ for all $k \in [q]$ as \EMPH{flagged $\mathtt{W}$-elements}
    and denote this collection as a multiset $A_W'$; we refer to the collection 
    of $\gamma_{k,\ell}$ with $2\leq \ell\leq d(w_k)$ for all $k \in [q]$ as 
    \EMPH{unflagged $\mathtt{W}$-elements} and denote this collection as a multiset $A_W$.

    \item (\EMPH{The $\mathtt{E}$-elements}) For every edge
    $(u_i,v_j,w_k)\in E$, we create an integer $h_{(i,j,k)}$ defined as:
    \[h_{(i,j,k)}=M^8+M^7-kM^2-jM-i.\]
    We call these the \EMPH{$\mathtt{E}$-elements} and denote their collection
    by the multiset $H$.
    \item (\EMPH{The $\mathtt{UV}$-connectors}) For every edge $(u_i,v_j,w_k)\in E$,
    we create two integers $x_{(i,j,k)}$ and $x'_{(i,j,k)}$ defined as:
    \[x_{(i,j,k)}=M^8-M^5-M^4-jM-i
      \quad \text{and} \quad
      x'_{(i,j,k)}=M^8-M^5-M^4-2M^3-jM-i.\]
    We call these the \EMPH{$\mathtt{UV}$-connectors}. We refer to the collection
    of $x'_{(i,j,k)}$ for all $(u_i,v_j,w_k)\in E$ as
    \EMPH{flagged $\mathtt{UV}$-connectors} and denote this collection as a multiset $X'$;
    we refer to the collection of $x_{(i,j,k)}$ for all
    $(u_i,v_j,w_k)\in E$ as \EMPH{unflagged $\mathtt{UV}$-connectors} and denote
    this collection as a multiset $X$.
    \item (\EMPH{The $\mathtt{WE}$-connectors}) For every edge $(u_i,v_j,w_k)\in E$,
    we create two integers $y_{(i,j,k)}$ and $y'_{(i,j,k)}$ defined as:
    \[y_{(i,j,k)}=M^8-M^7-M^6+jM+i
      \quad \text{and} \quad
      y'_{(i,j,k)}=M^8-M^7-M^6+2M^3+jM+i.\]
    We call these the \EMPH{$\mathtt{WE}$-connectors}. We refer to the collection
    of $y'_{(i,j,k)}$ for all $(u_i,v_j,w_k)\in E$ as
    \EMPH{flagged $\mathtt{WE}$-connectors} and denote this collection as a multiset $Y'$;
    we refer to the collection of $y_{(i,j,k)}$ for all
    $(u_i,v_j,w_k)\in E$ as \EMPH{unflagged $\mathtt{WE}$-connectors} and denote
    this collection as a multiset $Y$.
    \item (\EMPH{The padding elements}) For every edge
    $(u_i,v_j,w_k)\in E$, we create an integer $p_{(i,j,k)}$ defined as:
    \[p_{(i,j,k)}=M^8+M^7+M^6+M^5+M^4.\]
    We call these the \EMPH{padding elements} and denote their collection
    by the multiset $P$.
  \end{itemize}
  The multiset of input integers comprises all these elements:
\begin{equation}
    \underbrace{A_U \cup A_U'}_{ \text{$\mathtt{U}$-elements} }  \cup \underbrace{ A_V \cup A_V' }_{ \text{$\mathtt{V}$-elements} } \cup \underbrace{ A_W \cup A_W' }_{ \text{$\mathtt{W}$-elements} } \cup \underbrace{ H }_{ \text{$\mathtt{E}$-elements} }\cup  \underbrace{X \cup X'}_{ \text{$\mathtt{UV}$-connectors} } \cup   \underbrace{Y \cup Y' }_{ \text{$\mathtt{WE}$-connectors} } \cup \underbrace{P}_{ \text{padding} }
\end{equation}
  Finally, we set $B=3M^8$. See \cref{tab:3Part-instance} for a base-$M$
  representation of these numbers.

  \begin{table}[htbp]
  \centering
  \begingroup
  \scriptsize
  \setlength{\tabcolsep}{2pt}
  \renewcommand{\arraystretch}{1.25}
  \begin{tabular}{@{}lccr|*{9}{c}@{}}
  \hline
  Name & Multiset & Index range & Item notation
       & $M^8$ & $M^7$ & $M^6$ & $M^5$ & $M^4$
       & $M^3$ & $M^2$ & $M^1$ & $M^0$\\
  \hline
  Flagged $\mathtt{U}$-elements
    & $A_U'$ & $\forall u_i\in U$ & $\alpha_{i,1}$
    & 1 & 0 & 0 & 0 & 1 & 1 & 0 & 0 & $i$\\
  Unflagged $\mathtt{U}$-elements
    & $A_U$ & $\forall u_i\in U,\ \ell\in\{2,\ldots,d(u_i)\}$ & $\alpha_{i,\ell}$
    & 1 & 0 & 0 & 0 & 1 & 0 & 0 & 0 & $i$\\
  \hline
  Flagged $\mathtt{V}$-elements
    & $A_V'$ & $\forall v_j\in V$ & $\beta_{j,1}$
    & 1 & 0 & 0 & 1 & 0 & 1 & 0 & $j$ & 0\\
  Unflagged $\mathtt{V}$-elements
    & $A_V$ & $\forall v_j\in V,\ \ell\in\{2,\ldots,d(v_j)\}$ & $\beta_{j,\ell}$
    & 1 & 0 & 0 & 1 & 0 & 0 & 0 & $j$ & 0\\
  \hline
  Flagged $\mathtt{W}$-elements
    & $A_W'$ & $\forall w_k\in W$ & $\gamma_{k,1}$
    & 1 & 0 & 1 & 0 & 0 & $-2$ & $k$ & 0 & 0\\
  Unflagged $\mathtt{W}$-elements
    & $A_W$ & $\forall w_k\in W,\ \ell\in\{2,\ldots,d(w_k)\}$ & $\gamma_{k,\ell}$
    & 1 & 0 & 1 & 0 & 0 & 0 & $k$ & 0 & 0\\
  \hline
  $\mathtt{E}$-elements & $H$ & $\forall (u_i,v_j,w_k)\in E$ & $h_{(i,j,k)}$
    & 1 & 1 & 0 & 0 & 0 & 0 & $-k$ & $-j$ & $-i$\\
  \hline
  Flagged $\mathtt{UV}$-connectors & $X'$ & $\forall (u_i,v_j,w_k)\in E$ & $x'_{(i,j,k)}$
    & 1 & 0 & 0 & $-1$ & $-1$ & $-2$ & 0 & $-j$ & $-i$\\
  Unflagged $\mathtt{UV}$-connectors & $X$ & $\forall (u_i,v_j,w_k)\in E$ & $x_{(i,j,k)}$
    & 1 & 0 & 0 & $-1$ & $-1$ & 0 & 0 & $-j$ & $-i$\\
  \hline
  Flagged $\mathtt{WE}$-connectors & $Y'$ & $\forall (u_i,v_j,w_k)\in E$ & $y'_{(i,j,k)}$
    & 1 & $-1$ & $-1$ & 0 & 0 & 2 & 0 & $j$ & $i$\\
  Unflagged $\mathtt{WE}$-connectors & $Y$ & $\forall (u_i,v_j,w_k)\in E$ & $y_{(i,j,k)}$
    & 1 & $-1$ & $-1$ & 0 & 0 & 0 & 0 & $j$ & $i$\\
  \hline
  Padding elements & $P$ & $\forall (u_i,v_j,w_k)\in E$ & $p_{(i,j,k)}$
    & 1 & 1 & 1 & 1 & 1 & 0 & 0 & 0 & 0\\
  \hline
   Bin size & $B$ &  & 
    & 3 & 0 & 0 & 0 & 0 & 0 & 0 & 0 & 0\\
    \hline
  \end{tabular}
  \endgroup
  \caption{Coefficients of the powers of $M$ in the reduced
  \ThreePartition{} instance. For example, the second row represents
  $\alpha_{i,\ell}=M^8+M^4+i$, for $2\leq \ell\leq d(u_i)$. }
  \label{tab:3Part-instance}
  \end{table}
  
  We first check that the constructed integers form a valid
  \ThreePartition{} instance and that their size remains polynomial.
  \begin{claim}[Construction]\label{clm:new-partition-construction}
  The construction produces $n=9m$ items satisfying the numerical
  promises of \ThreePartition{}: each item is an integer in $(B/4,B/2)$ and the total sum of all items is $nB/3$, where $B\leq n^{\Oh(1)}$. The total binary encoding length is $\Oh(n\log n)$.
  \end{claim}
  \begin{proof}[Proof of the claim]
  \renewcommand{\qedsymbol}{$\blacksquare$}
  Each edge contributes one item from each of $A_U\cup A_U'$,
  $A_V\cup A_V'$, and $A_W\cup A_W'$, as well as one item from $H$,
  two from $X\cup X'$, two from $Y\cup Y'$, and one from $P$.
  Thus there are $9m$ items. Since $M\geq20$, every item differs
  from $M^8$ by less than $2M^7$ in absolute value. Hence every
  item lies strictly between $B/4$ and $B/2$.

  Next we compute the total sum. The coefficients of $M^0,M^1,M^2$ cancel in the total sum:
  the coefficients $i,j,k$ in the $\mathtt{U}$-elements, $\mathtt{V}$-elements, and $\mathtt{W}$-elements cancel
  those in $H$, while the two $\mathtt{UV}$-connectors cancel the two
  $\mathtt{WE}$-connectors edge by edge. At each of $M^4,M^5,M^6,M^7$,
  the total coefficient is $m-2m+m=0$. At $M^3$,
  the flagged $\mathtt{U}$-elements, flagged $\mathtt{V}$-elements,
  and flagged $\mathtt{W}$-elements contribute $q+q-2q=0$;
  the flagged $\mathtt{UV}$-connectors and flagged
  $\mathtt{WE}$-connectors contribute $-2m+2m=0$:
  \[
    q+q-2q-2m+2m=0.
  \]
  Each of the $9m$ items has coefficient $1$ at $M^8$.
  Their total is therefore $9mM^8=nB/3$.
  Finally, no node is isolated, so $q\leq m$ and $M=\Oh(m)$.
  Thus $B=3M^8$ is polynomial in $n$ and the output has
  encoding length $\Oh(n\log n)$.
  \end{proof}

  \subparagraph*{Completeness.}
  For each edge in a perfect matching, we group: (i) its flagged
  $\mathtt{U}$-element and flagged $\mathtt{V}$-element with its
  flagged $\mathtt{UV}$-connector, and (ii) its flagged $\mathtt{W}$-element
  and $\mathtt{E}$-element with its flagged $\mathtt{WE}$-connector, and (iii) the remaining $\mathtt{UV}$-connector and $\mathtt{WE}$-connector with a padding element. For each edge outside the matching,
  we instead group: (i) an unflagged $\mathtt{U}$-element, an unflagged
  $\mathtt{V}$-element, together with an unflagged
  $\mathtt{UV}$-connector in the first triple, and (ii) an unflagged $\mathtt{W}$-element, an $\mathtt{E}$-element, together with an unflagged $\mathtt{WE}$-connector in the second triple, and (iii) a flagged $\mathtt{UV}$-connector, a flagged $\mathtt{WE}$-connector, together with a padding element in the third triple. 
  \begin{claim}\label{clm:new-partition-completeness}
  If $E$ has a perfect matching, then the constructed items admit
  a partition into triples of sum $B$.
  \end{claim}
  \begin{proof}[Proof of the claim]
  \renewcommand{\qedsymbol}{$\blacksquare$}
  Let $E'\subseteq E$ be a perfect matching, so $|E'|=q$.
  For each edge $e=(u_i,v_j,w_k)\in E'$, form the three triples
  \[
    \{\alpha_{i,1},\beta_{j,1},x'_{(i,j,k)}\},
    \qquad
    \{\gamma_{k,1},h_{(i,j,k)},y'_{(i,j,k)}\},
    \qquad 
    \{x_{(i,j,k)}, y_{(i,j,k)}, p_{(i,j,k)}\}.
  \]

  For each edge $(u_i,v_j,w_k)\in E\setminus E'$, assign one
  unused element $\alpha_{i,\ell_1}\in A_U$,
  $\beta_{j,\ell_2}\in A_V$, and $\gamma_{k,\ell_3}\in A_W$
  with $\ell_1,\ell_2,\ell_3\ge 2$ to that edge. Such an assignment
  is possible because each node has exactly one incident edge in
  $E'$, leaving its $d(a)-1$ unflagged elements for its other
  $d(a)-1$ incident edges. Then, we form the triples
  \[
    \{\alpha_{i,\ell_1},\beta_{j,\ell_2},x_{(i,j,k)}\},
    \qquad
    \{\gamma_{k,\ell_3},h_{(i,j,k)},y_{(i,j,k)}\},
    \qquad 
    \{x'_{(i,j,k)}, y'_{(i,j,k)}, p_{(i,j,k)}\}.
  \]
  
  Every item is used once. In each triple
  every coefficient of $M^0$ to $M^7$ sums to zero by
  \cref{tab:3Part-instance}; while only the coefficients
  $1$ of $M^8$ add up, making each triple sum up to $B=3M^8$.
  \end{proof}

  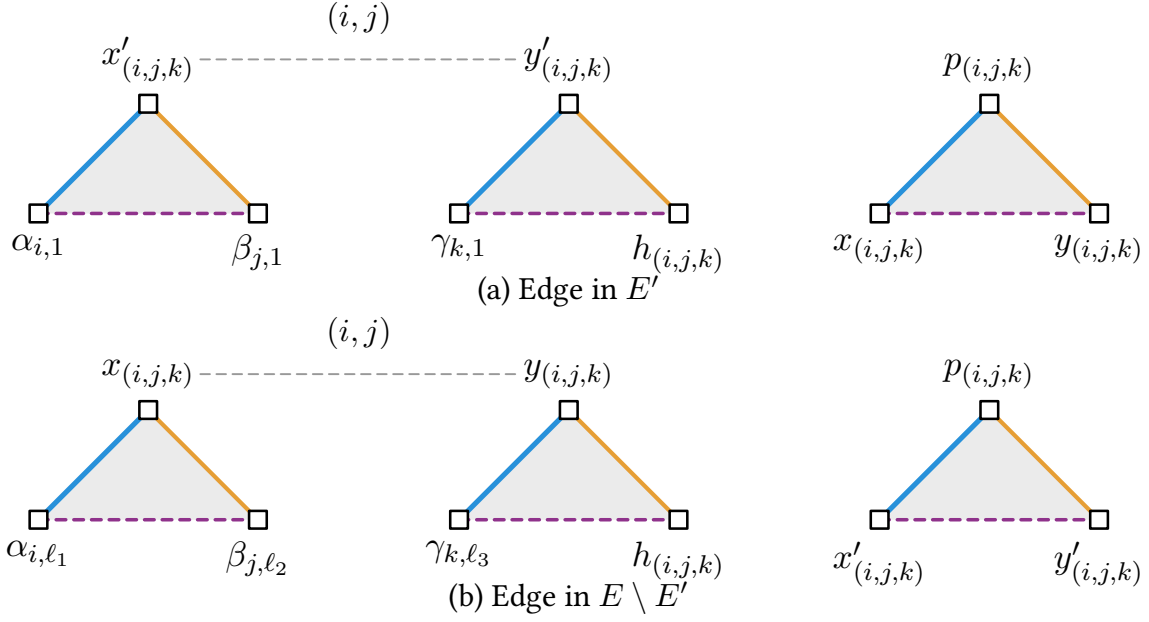
\begin{figure}[!htb]
  \centering
  \begingroup
  \definecolor{PMRblue}{RGB}{40,145,218}
  \definecolor{PMRpurple}{RGB}{143,50,139}
  \definecolor{PMRorange}{RGB}{230,158,53}
  \definecolor{PMRgray}{RGB}{235,235,235}

  \resizebox{.92\linewidth}{!}{%
  \begin{tikzpicture}[
    x=1cm,y=1cm,font=\normalsize,
    line cap=round,line join=round,
    blue path/.style={draw=PMRblue,line width=1.6pt},
    purple path/.style={draw=PMRpurple,line width=1.1pt,
      dash pattern=on 3pt off 2.5pt},
    orange path/.style={draw=PMRorange,line width=1.3pt},
    item vertex/.style={rectangle,draw=black,fill=white,
      line width=.8pt,minimum size=5.5pt,inner sep=0pt},
    item label/.style={inner sep=1.5pt},
    guide/.style={draw=black!40,densely dashed,line width=.6pt},
    panel label/.style={font=\small,align=center}
  ]
    \foreach \row in {0,1}{
      \begin{scope}[yshift=-3.35*\row cm]
        \foreach \col in {0,1,2}{
          \begin{scope}[xshift=4.6*\col cm]
            \fill[PMRgray] (0,0)--(2.4,0)--(1.2,1.2)--cycle;
            \draw[blue path] (0,0)--(1.2,1.2);
            \draw[purple path] (0,0)--(2.4,0);
            \draw[orange path] (1.2,1.2)--(2.4,0);
            \foreach \point in {(0,0),(2.4,0),(1.2,1.2)}{
              \node[item vertex] at \point {};
            }
          \end{scope}
        }
      \end{scope}
    }

    \node[item label,above=5pt] (top-x) at (1.2,1.2)
      {$x'_{(i,j,k)}$};
    \node[item label,below=5pt] at (0,0) {$\alpha_{i,1}$};
    \node[item label,below=5pt] at (2.4,0) {$\beta_{j,1}$};
    \node[item label,above=5pt] (top-y) at (5.8,1.2)
      {$y'_{(i,j,k)}$};
    \node[item label,below=5pt] at (4.6,0) {$\gamma_{k,1}$};
    \node[item label,below=5pt] at (7,0) {$h_{(i,j,k)}$};
    \draw[guide] (top-x.east)--(top-y.west)
      node[midway,above=4pt,font=\small] {$(i,j)$};
    \node[item label,above=5pt] at (10.4,1.2) {$p_{(i,j,k)}$};
    \node[item label,below=5pt] at (9.2,0) {$x_{(i,j,k)}$};
    \node[item label,below=5pt] at (11.6,0) {$y_{(i,j,k)}$};
    \node[panel label] at (5.8,-.85) {(a) Edge in $E'$};

    \begin{scope}[yshift=-3.35cm]
      \node[item label,above=5pt] (bottom-x) at (1.2,1.2)
        {$x_{(i,j,k)}$};
      \node[item label,below=5pt] at (0,0) {$\alpha_{i,\ell_1}$};
      \node[item label,below=5pt] at (2.4,0) {$\beta_{j,\ell_2}$};
      \node[item label,above=5pt] (bottom-y) at (5.8,1.2)
        {$y_{(i,j,k)}$};
      \node[item label,below=5pt] at (4.6,0) {$\gamma_{k,\ell_3}$};
      \node[item label,below=5pt] at (7,0) {$h_{(i,j,k)}$};
      \draw[guide] (bottom-x.east)--(bottom-y.west)
        node[midway,above=4pt,font=\small] {$(i,j)$};
      \node[item label,above=5pt] at (10.4,1.2) {$p_{(i,j,k)}$};
      \node[item label,below=5pt] at (9.2,0) {$x'_{(i,j,k)}$};
      \node[item label,below=5pt] at (11.6,0) {$y'_{(i,j,k)}$};
      \node[panel label] at (5.8,-.85) {(b) Edge in $E\setminus E'$};
    \end{scope}
  \end{tikzpicture}%
  }
  \endgroup
  \caption{The three triples associated with an edge
  $(u_i,v_j,w_k)\in E'$ (top) and an edge
  $(u_i,v_j,w_k)\in E\setminus E'$ (bottom), where
  $\ell_1,\ell_2,\ell_3\geq 2$.
  Each shaded triangle has sum $B$. The gray dashed guides indicate
  that the $\mathtt{UV}$-connector and $\mathtt{WE}$-connector have
  the same indices $i,j$.}
  \label{fig:new-matching-partition}
  \end{figure}

  \subparagraph*{Soundness.}
  We now obtain a perfect matching from any solution of the 
  \ThreePartition{} instance. The high powers of $M$ determine the three
  possible types of triples, and the lower powers determine
  the values of $i,j,k$ and whether their elements are flagged.
  \begin{claim}\label{clm:new-partition-soundness}
  If the constructed items admit a partition into triples of
  sum $B$, then $E$ has a perfect matching.
  \end{claim}
  \begin{proof}[Proof of the claim]
  \renewcommand{\qedsymbol}{$\blacksquare$}
  Consider such a solution to the \ThreePartition{} instance. In any triple, the absolute value
  of each coefficient of $M^0$ to $M^7$ is less
  than $M$: the exact bounds on the absolute values are $3q$ for $M^0,M^1,M^2$, $6$ for
  $M^3$, and $3$ for $M^4,\ldots,M^7$.
  Thus equality to $B$ can be checked by simply checking if coefficients
  of $M^0$ to $M^7$ sum up to $0$.

  We first inspect the $M^7,M^6,M^5,M^4$ coefficients (see 
  \cref{tab:3Part-instance}). The only ways to form triples that 
  cancel these coefficients are
  \begin{itemize}
    \item a $\mathtt{U}$-element, a $\mathtt{V}$-element, and a $\mathtt{UV}$-connector;
    \item a $\mathtt{W}$-element, an $\mathtt{E}$-element, and a $\mathtt{WE}$-connector; 
    \item or a $\mathtt{UV}$-connector, a $\mathtt{WE}$-connector, and a padding element.
  \end{itemize}
  Indeed, a padding element
  requires a $\mathtt{UV}$-connector and a $\mathtt{WE}$-connector to cancel its four
  positive coefficients. Without a padding element, a $\mathtt{UV}$-connector
  requires a $\mathtt{U}$-element and a $\mathtt{V}$-element, while a $\mathtt{WE}$-connector
  requires a $\mathtt{W}$-element and an $\mathtt{E}$-element. The
  $M^3$ coefficient further restricts the possibilities. Thus every
  triple consists of one item from each of exactly one
  of the following six combinations of multisets:
  \[
  \begin{array}{lll}
    A_U\times A_V\times X, & A_W\times H\times Y, & X\times Y\times P,\\
    A_U'\times A_V'\times X', & A_W'\times H\times Y', & X'\times Y'\times P.
  \end{array}
  \]
  Consider a triple
  $(\alpha_{i,1},\beta_{j,1},x'_{(i',j',k')})\in
  A_U'\times A_V'\times X'$ or a triple
  $(\alpha_{i,\ell_1},\beta_{j,\ell_2},x_{(i',j',k')})\in
  A_U\times A_V\times X$. Its $M^0$ coefficient sums to
  $i-i'$, and its $M^1$ coefficient sums to $j-j'$.
  Hence $i=i'$ and $j=j'$ in either case. Similarly, consider a triple
  $(\gamma_{k,1},h_{(i,j,k')},y'_{(i',j',k'')})\in
  A_W'\times H\times Y'$ or a triple
  $(\gamma_{k,\ell_3},h_{(i,j,k')},y_{(i',j',k'')})\in
  A_W\times H\times Y$. Its $M^2$, $M^1$, and $M^0$
  coefficients sum to $k-k'$, $j'-j$, and $i'-i$,
  respectively. Thus $k=k'$, $j=j'$, and $i=i'$. Finally, consider a triple
  $(x_{(i,j,k)},y_{(i',j',k')},p_{(i'',j'',k'')})\in
  X\times Y\times P$ or a triple
  $(x'_{(i,j,k)},y'_{(i',j',k')},p_{(i'',j'',k'')})\in
  X'\times Y'\times P$. Its $M^0$ and $M^1$
  coefficients sum to $i'-i$ and $j'-j$, respectively.
  Hence $i=i'$ and $j=j'$ in either case as well.

  Consequently, every triple in the solution of the \ThreePartition{} 
  instance has one of the following forms:
  \begin{align*}
    (\alpha_{i,1},\beta_{j,1},x'_{(i,j,k)})
      &\in A_U'\times A_V'\times X',\\
    (\alpha_{i,\ell_1},\beta_{j,\ell_2},x_{(i,j,k)})
      &\in A_U\times A_V\times X,\\
    (\gamma_{k,1},h_{(i,j,k)},y'_{(i,j,k')})
      &\in A_W'\times H\times Y',\\
    (\gamma_{k,\ell_3},h_{(i,j,k)},y_{(i,j,k')})
      &\in A_W\times H\times Y,\\
    (x_{(i,j,k)},y_{(i,j,k')},p_{(i'',j'',k'')})
      &\in X\times Y\times P,\\
    (x'_{(i,j,k)},y'_{(i,j,k')},p_{(i'',j'',k'')})
      &\in X'\times Y'\times P,
  \end{align*}
   where $\ell_1,\ell_2,\ell_3\ge 2$.

  Fix $i,j\in[q]$. For each $k$ with $(u_i,v_j,w_k)\in E$,
  the construction supplies $x_{(i,j,k)}\in X$ and
  $y_{(i,j,k)}\in Y$, as well as $x'_{(i,j,k)}\in X'$ and
  $y'_{(i,j,k)}\in Y'$. The number of items of each of these
  four forms is $|\{k\in[q]:(u_i,v_j,w_k)\in E\}|$.
  Every triple in $X\times Y\times P$ containing an item
  $x_{(i,j,k)}$ also contains an item $y_{(i,j,k')}$, and
  conversely. The analogous statement holds for triples in
  $X'\times Y'\times P$ containing $x'_{(i,j,k)}$ and
  $y'_{(i,j,k')}$. Since every item occurs in exactly one triple,
  the number of triples in $A_U\times A_V\times X$ containing
  $x_{(i,j,k)}$ for some $k$ equals the number of triples in
  $A_W\times H\times Y$ containing $y_{(i,j,k')}$ for some
  $k'$. Likewise, the number of triples in
  $A_U'\times A_V'\times X'$ containing $x'_{(i,j,k)}$ equals
  the number of triples in $A_W'\times H\times Y'$ containing
  $y'_{(i,j,k')}$. Choose a bijection $\pi^{(i,j)}_1$ from the triples counted
  in $A_U\times A_V\times X$ to those counted in
  $A_W\times H\times Y$, and another bijection $\pi^{(i,j)}_2$ from the triples counted
  in $A_U'\times A_V'\times X'$ to those counted in
  $A_W'\times H\times Y'$.

  Under the bijection $\pi^{(i,j)}_1$, paired triples have the forms
  $(\alpha_{i,\ell_1},\beta_{j,\ell_2},x_{(i,j,k')})$ and
  $(\gamma_{k,\ell_3},h_{(i,j,k)},y_{(i,j,k'')})$.
  Under the bijection $\pi^{(i,j)}_2$, paired triples have the forms
  $(\alpha_{i,1},\beta_{j,1},x'_{(i,j,k')})$ and
  $(\gamma_{k,1},h_{(i,j,k)},y'_{(i,j,k'')})$.
  In either case, the $\mathtt{U}$-element represents $u_i$,
  the $\mathtt{V}$-element represents $v_j$, and the
  $\mathtt{W}$-element represents $w_k$. Moreover,
  $(u_i,v_j,w_k)\in E$ because the $\mathtt{E}$-element is
  $h_{(i,j,k)}$.

  For each $i,j,k\in[q]$, the items $\alpha_{i,1}$,
  $\beta_{j,1}$, and $\gamma_{k,1}$ are the unique flagged
  elements representing $u_i$, $v_j$, and $w_k$, respectively.
  Each occurs in exactly one of the pairs of triples under the
  second bijection. The $\mathtt{E}$-elements in those pairs
  represent distinct edges, since each item of $H$ is used
  exactly once. These edges therefore cover every node of
  $U\cup V\cup W$ exactly once and form a perfect matching in $E$.
  \end{proof}

  The three claims establish the reduction.
\end{proof}

\begin{proof}[Proof of \cref{thm:ETH-3Partition}]
Let $m$ be the number of edges in a \textsc{3-Dimensional Matching}
instance. By \cref{lem:matching-partition}, the reduction produces
$n=9m$ items, so the number of triples sought in the resulting
\ThreePartition{} instance is $n/3=3m$. Moreover,
$B\leq n^{\Oh(1)}=m^{\Oh(1)}$.

Suppose that \ThreePartition{} on $3r$ numbers with
$B=r^{\Oh(1)}$ could be solved in time
$2^{o(r)}|B|^{\Oh(1)}$, where $|B|$ is the binary encoding
length of $B$. Applying that algorithm with $r=3m$ after the
polynomial-time reduction would decide \textsc{3-Dimensional Matching}
in time $2^{o(m)}|I|^{\Oh(1)}$, where $|I|$ is the source input
length. This contradicts
\cite[Theorem~3.1]{Jansen2016}
under ETH.
\end{proof}

\end{document}